\newcommand{\longversion}[1]{}
\newcommand{\shortversion}[1]{#1}
\documentclass[a4paper,UKenglish,cleveref, autoref, thm-restate]{lipics-v2021}
\title{On Parameterized Complexity of Binary Networked Public Goods Game\footnote{To appear as full paper in \textbf{AAMAS 2022}}} %TODO Please add

\author{Arnab Maiti}{Indian Institute of Technology Kharagpur }{arnabmaiti@iitkgp.ac.in}{https://orcid.org/0000-0002-9142-6255}{}

\author{Palash Dey}{Indian Institute of Technology Kharagpur }{palash.dey@cse.iitkgp.ac.in}{https://orcid.org/0000-0003-0071-9464}{}

\authorrunning{A.\,Maiti and P.\,Dey} %TODO mandatory. First: Use abbreviated first/middle names. Second (only in severe cases): Use first author plus 'et al.'

\Copyright{Arnab Maiti and Palash Dey} %TODO mandatory, please use full first names. LIPIcs license is "CC-BY";  http://creativecommons.org/licenses/by/3.0/

 %TODO mandatory: Please choose ACM 2012 classifications from https://dl.acm.org/ccs/ccs_flat.cfm 

 %TODO mandatory; please add comma-separated list of keywords

\category{} %optional, e.g. invited paper

\relatedversion{} %optional, e.g. full version hosted on arXiv, HAL, or other respository/website
\EventEditors{John Q. Open and Joan R. Access}
\EventNoEds{2}
\EventLongTitle{42nd Conference on Very Important Topics (CVIT 2016)}
\EventShortTitle{CVIT 2016}
\EventAcronym{CVIT}
\EventYear{2016}
\EventDate{December 24--27, 2016}
\EventLocation{Little Whinging, United Kingdom}
\EventLogo{}
\SeriesVolume{42}
\ArticleNo{23}
\usepackage{xspace}
\usepackage{cleveref}
\usepackage{amsthm}
\newcommand{\NP}{\ensuremath{\mathsf{NP}}\xspace}
\newcommand{\NPC}{\ensuremath{\mathsf{NP}}\text{-complete}\xspace}
\newcommand{\NPH}{\ensuremath{\mathsf{NP}}\text{-hard}\xspace}
\newcommand{\PNPH}{para-\ensuremath{\mathsf{NP}\text{-hard}}\xspace}
\newcommand{\el}{\ensuremath{\ell}\xspace}

\newcommand{\WOH}{\ensuremath{\mathsf{W[1]}}-hard\xspace}

\newcommand{\WO}{\ensuremath{\mathsf{W[1]}}\xspace}
\newcommand{\WT}{\ensuremath{\mathsf{W[2]}}\xspace}

\newcommand{\WTH}{\ensuremath{\mathsf{W[2]}}-hard\xspace}
\newcommand{\FPT}{\ensuremath{\mathsf{FPT}}\xspace}
\newcommand{\fpt}{\ensuremath{\mathsf{FPT}}\xspace}
\newcommand{\xp}{\ensuremath{\mathsf{XP}}\xspace}
\newcommand{\tsat}{\ensuremath{(3,\text{B}2)}-{\sc SAT}\xspace}

\newcommand{\GFAC}{{\sc General Factor}\xspace}

\newcommand{\psne}{{\sc Exists-PSNE}\xspace}
\newcommand{\bnpg}{\psne}

\newcommand{\vc}{\text{vc(\ensuremath{\GG})}}

\newcommand{\ds}{{\sc Dominating Set}\xspace}

\let\oldlambda\lambda
\renewcommand{\lambda}{\ensuremath{\oldlambda}\xspace}
\let\oldalpha\alpha
\renewcommand{\alpha}{\ensuremath{\oldalpha}\xspace}
\let\oldDelta\Delta
\renewcommand{\Delta}{\ensuremath{\oldDelta}\xspace}

\newcommand{\YES}{{\sc yes}\xspace}

\newcommand{\yes}{{\sc yes}\xspace}
\newcommand{\no}{{\sc no}\xspace}
\newcommand{\true}{\text{{\sc true}}\xspace}
\newcommand{\false}{\text{{\sc false}}\xspace}

\newcommand{\CC}{\ensuremath{\mathcal C}\xspace}

\newcommand{\EE}{\ensuremath{\mathcal E}\xspace}
\newcommand{\FF}{\ensuremath{\mathcal F}\xspace}
\newcommand{\GG}{\ensuremath{\mathcal G}\xspace}
\newcommand{\HH}{\ensuremath{\mathcal H}\xspace}

\newcommand{\OO}{\ensuremath{\mathcal O}\xspace}
\newcommand{\PP}{\ensuremath{\mathcal P}\xspace}

\renewcommand{\SS}{\ensuremath{\mathcal S}\xspace}
\newcommand{\TT}{\ensuremath{\mathcal T}\xspace}
\newcommand{\UU}{\ensuremath{\mathcal U}\xspace}
\newcommand{\VV}{\ensuremath{\mathcal V}\xspace}
\newcommand{\WW}{\ensuremath{\mathcal W}\xspace}
\newcommand{\XX}{\ensuremath{\mathcal X}\xspace}
\newcommand{\YY}{\ensuremath{\mathcal Y}\xspace}

\newcommand{\NB}{\ensuremath{\mathbb N}\xspace}
\newcommand{\RB}{\ensuremath{\mathbb R^+}\xspace}

\newcommand{\pr}{\ensuremath{\prime}}

\usepackage{tikz}
\usetikzlibrary{arrows,positioning}

\begin{document}
\nolinenumbers
\def\UrlBreaks{\do\/\do-}
\maketitle

\begin{abstract}
	In the Binary Networked Public Goods (BNPG for short) game, every player needs to decide if she participates in a public project whose utility is shared equally by the community. We study the problem of deciding if there exists a pure strategy Nash equilibrium (PSNE) in such games. The problem is already known to be \NPC. This casts doubt on predictive power of PSNE in BNPG games. We provide fine-grained analysis of this problem under the lens of parameterized complexity theory. We consider various natural graph parameters and show \WO-hardness, XP, and \fpt results. Hence, our work significantly improves our understanding of BNPG games where PSNE serves as a reliable solution concept. We finally prove that some graph classes, for example path, cycle, bi-clique, and complete graph, always have a PSNE if the utility function of the players are same.
\end{abstract}
\section{Introduction}

In a public goods game, players need to decide if they contribute in a public project and, if yes, then how much. The outcome of such public projects is typically shared equally by all the players. Public goods games are effective in modeling tension between individual cost vs community well beings~\cite{kollock1998social,santos2008social}. One of the well-explored variants of the above game is the {\em networked} public goods game where we assume a network structure on the players and the utilities of individual players depend on the action of them and their neighbors only~\cite{bramoulle2007public}.

An important class of networked public goods game is the {\em binary} networked public goods (BNPG for short) game where players only need to decide if they participate (play $1$) in the public project or not (play $0$)~\cite{galeotti2010network}. Although this seems restricted, such games are still powerful enough to model various important real world application scenarios. For some motivating examples, let us think of an air-borne virus pandemic like Covid-19 where individuals need to decide whether to wear a mask or not. While individuals may feel uncomfortable while wearing a mask, the benefits of herd immunity, if achieved by a large fraction of population wearing a mask, will be shared by the entire community. Indeed, there are reports that a considerable fraction of population refuse to wear a mask during Covid-19 pandemic~\cite{nomask,nomask1}. Another important application is whether to report a crime or not. While individuals who report crimes may be at risk, the benefit of having lower crime rates will be enjoyed by the entire community. The general observation at many places is that crimes are often under-reported~\cite{crimereport}.

Computing a pure strategy Nash equilibrium (PSNE) in any game is a fundamental question. The concept of Nash equilibrium guides social planner to predict how players will act in a strategic setting and act accordingly. We know that the \psne problem, where we are asked to decide if a BNPG game has a PSNE, is \NPC~\cite{yu2020computing}. In this work, we provide a comprehensive study of the parameterized complexity of the \psne problem.
   
%Finally we consider fully homogeneous BNPG games on special graph classes. We show that a PSNE always exists for a path [\Cref{fully-path}], complete graph\longversion{ [\Cref{fully-complete}]}, cycle\longversion{ [\Cref{fully-cycle}]}, and bi-clique\longversion{ [\Cref{fully-biclique}]}\shortversion{ [\Cref{thm:rest}]}.

%We believe that our fine-grained complexity theoretic study of \psne enables us to understand the usefulness of pure strategy Nash equilibrium for BNPG games for various graph classes. PSNE may be a reliable solution concept for predicting the behavior of the players in the situations where the problem admits an \FPT algorithm. However, it will be better to use some other solution concepts in the settings where \psne is intractable; players will be unable to even find a PSNE in these cases let alone playing according to it.

\subsection{Related Work} The immediate predecessor of our work is \cite{yu2020computing} where the authors initiate the algorithmic question of \psne. Our work broadly belongs to the field of graphical games where there is a graph structure on the players and a player's utility depends only on the actions of her neighbors~\cite{kearns2007graphical}. A central question in graphical games is to find complexity of the problem of computing an equilibrium~\cite{elkind2006nash,daskalakis2009complexity,gottlob2005pure}. Network public goods games are a special case of graphical games where the utility of players depends only on the sum of the ``efforts'' put in by neighbors and the cost of her action. Many models of the network public goods game have been explored which are fine-tuned to different applications. Important examples of such applications include economics, research collaboration, social influence, etc.~\cite{burt1987social,valente1995network,conley2010learning,valente2005network}. The BNPG model is closely related to that proposed in Bramoull{\'e} et al. \cite{bramoulle2007public}. There are however two qualitative distinctions (a) Bramoull{\'e} et al. focus on the continuous investment model whereas BNPG model focuses on binary investment decisions and (b) Bramoull{\'e} et al. assume homogeneous concave utilities whereas BNPG model considers a more general setting. Supermodular network games \cite{manshadi2009supermodular} and best-shot games (which is actually a special case of BNPG game) \cite{dall2011optimal}, etc. \cite{galeotti2010network,komarovsky2015efficient,levit2018incentive} are other important variations of graphical games. In the model of Supermodular network games, each agent’s payoff is a function of the aggregate action of its neighbors and it exhibits strategic complementarity. An important example of supermodular games on graphs are technology adoption games which have been studied in the social network literature \cite{kleinberg2007cascading,immorlica2007role,morris2000contagion}.

\subsection{Parameters}
As \psne is \NPC~\cite{yu2020computing},  we provide a comprehensive study of the parameterized complexity of the \psne problem w.r.t. the following parameters:
\begin{itemize}
	\item \textit{Maximum Degree: }Many applications of BNPG games involve human beings as nodes in the network. Due to human cognitive limitation, such graphs often exhibit small maximum degree. With this motivation, we consider the maximum degree of the graph as our parameter. %Scale-free networks have been proved to be quite successful in modeling real-world networks. These networks have small average degree although there could be a few nodes with large degree. 
	\item \textit{Diameter:} Graphs which involve human beings as nodes tend to have a small diameter. Therefore, we consider the diameter of the graph as our parameter.
	\item \textit{Distance from tree and complete graph:} Trees and complete graphs are important classes of graphs in the context of BNPG games. It is already known from previous work that \psne is polynomial-time solvable for trees and complete graphs. Therefore the next natural question would be to check the tractibility of those instances where the graphs are quite close to being a tree or a complete graph. For this purpose, we consider the parameters distance from tree, which is also known as circuit rank, and distance from complete graphs [\Cref{def:d1}]. 
	\item \textit{Treedepth and Treewidth:} We also consider treedepth and treewidth as parameters as they have often turned out to be useful parameters to obtain a fixed-parameter-tractable (\FPT) algorithm for many classical problems for which it is known that the problem is polynomial-time solvable for trees. 
	\item \textit{Number of participating and non-participating players:} One may wish to know what are the equilibria during a pandemic like Covid-19 example where most and least people wear masks. For such scenarios, the number of participating (who play $1$) and non-participating players (who play $0$) are the natural parameters.
\end{itemize}

\section{Preliminaries}\label{sec:prelim}

%\subsubsection*{Model:}
For a set \XX, we denote its power set by $2^\XX$. We denote the set $\{1,\ldots,n\}$ by $[n]$. For $2$ sets \XX and \YY, we denote the set of functions from \XX to \YY by $\YY^\XX$.

Let $\GG=(\VV,\EE)$ be an undirected graph with $n$ vertices. An edge between $u,v\in \VV$ is represented by $\{u,v\}$. In a graph \GG, we denote the degree of any vertex $v$ by $d(v)$. For a subset $\UU\subseteq\VV$ of vertices (respectively a subset $\FF\subseteq\EE$ of edges), we denote the subgraph induced by \UU (respectively \FF) by $\GG[\UU]$ (respectively $\GG[\FF]$). A Binary Networked Public Goods (BNPG for short) game can be defined on \GG as follows. The set of players is \VV. The strategy set of every player is $\{0,1\}$. We denote the number of neighbors of $w$ in \GG who play $1$ in the strategy profile $(x_v)_{v\in\VV}$ by $n_w$; that is, $n_w=|\{u\in\VV:\{u,w\}\in\EE, x_u=1\}|$. For a strategy profile $(x_v)_{v\in\VV}\in\{0,1\}^{|\VV|}$, the utility $U_w((x_v)_{v\in\VV})$ of player (without abusing the notation much) $w\in\VV$ is defined as follows. 
\begin{equation*}
U_w((x_v)_{v\in\VV})=U_w(x_w,n_w)=g_w(x_w+n_w)-c_w\cdot x_w
\end{equation*}
where $g_w:\NB\cup\{0\}\longrightarrow\RB$ is a non-decreasing function in $x$ and $c_w\in\RB$ is a constant. We denote a BNPG game by $(\GG=(\VV,\EE),(g_v)_{v\in\VV},(c_v)_{v\in\VV})$. For any number $n\in\NB\cup\{0\}$ and function $g:\NB\cup\{0\}\longrightarrow\RB$, we define $\Delta g(n)=g(n+1)-g(n)$.
%A Binary Networked Public Goods (BNPG for short) game is characterized by a graph $\GG(\VV,\EE)$ where \VV=$\{v_1,v_2,\ldots,v_n\}$ is a set of players and \EE=$\{(u,v): u,v \in \EE$ and there is an edge between $u$ and $v\}$ , a binary strategy space $x_v \in\{0,1\}$ for each player $v\in \VV$ and a collection of utility functions $U_v(x_{v_1},\ldots,x_{v_n})$ for each player $v\in \VV$. We say that a player $v$ choses to invest iff $x_v=1$. Let $\forall v\in \VV$, $n_v$ denote the number of neighbours of $v$ choosing 1 as their strategy. Formally $n_v=|\{u: (u,v)\in\EE, x_u=1\}|$. We define player $v$'s utility function as follows:
%\begin{equation*}
%U_v(x_{v_1},\ldots,x_{v_n})=U_v(x_v,n_v)=g_v(x_v+n_v)-c_v\cdot x_v
%\end{equation*}
%where $g_v(x)$ is a non decreasing function in $x$. Note that we might change the notation during the proof of a specific theorem. So unless clearly specified we follow the notations defined above.  Note that the size of BNPG representation is $O(n^2)$ for a set of functions $g_v(.)$.
In general, every player $w\in\VV$ has a different mapping function $g_w(.)$ and hence we call this version of the game a \textit{ heterogeneous} BNPG game. If not mentioned otherwise, by BNPG game, we refer to a heterogeneous BNPG game. In this paper, we also study the following three special cases --- (i) \textit{ homogeneous:} $g_w=g$ for all $w\in\VV$, (ii) \textit{ fully homogeneous:} homogeneous and $c_w=c$ for all $w\in \VV$ and (iii) {strict:} for every player $w\in\VV$, we have $U_w(x_w=0,x_{-w})\ne U_w(x_w=1,x_{-w})$ for every strategy profile $x_{-w}$ of other players. So a BNPG game is strict if and only if $ \Delta g_w(k) \ne c_w, \;\forall w\in\VV, \forall k\in\{0,1,\ldots,d(w)\} $

A strategy profile $(x_v)_{v\in\VV}$ is called a \textit{ pure-strategy Nash Equilibrium (PSNE)} of a BNPG game if we have  $U_v(x_v,x_{-v}) \ge U_v(x^\pr_v, x_{-v}) \;\forall x^\pr_v\in\{0,1\}, \forall v\in\VV $. We call the problem of deciding if there exists a PSNE in BNPG games as \psne.

For a player $w$ in a BNPG game $(\GG=(\VV,\EE),(g_v)_{v\in\VV},\allowbreak(c_v)_{v\in\VV})$, we define her best response function $\beta_w:\{0,1,\allowbreak\ldots,n-1\}\longrightarrow2^{\{0,1\}}\setminus\{\emptyset\}$ as follows. For every $k\in\{0,1,\ldots,n-1\}$ and $a\in\{0,1\}$, we have $a\in\beta_w(k)$ if and only if, for every strategy profile $x_{-w}$ of players other than $w$ where exactly $k$ players in the neighborhood of $w$ play $1$, we have $U_w(x_w=a,x_{-w}) \ge U_w(x_w=a^\pr,x_{-w})$ for all $a^\pr\in\{0,1\}$. The following lemma proves that, for every function $\beta_w$, there is a function $g_w:\NB\cup\{0\}\longrightarrow\RB$ and constant $c_w$ such that $\beta_w$ is the best response function\longversion{; moreover such $g_w$ and $c_w$ can be computed in polynomial time}.

\begin{lemma}[$\star$]\label{lem:best-util}
	Let $\beta:\{0,1,\ldots,n-1\}\longrightarrow2^{\{0,1\}}\setminus\{\emptyset\}$ be an arbitrary function. Then we can compute in polynomial (in $n$) time a function $g:\NB\cup\{0\}\longrightarrow\RB$ and constant $c$ such that $\beta$ is the corresponding best response function.
\end{lemma}

We call a function $f:\NB\cup\{0\}\longrightarrow\RB$ {\em sub-additive} if $f(x+y)\le f(x)+f(y)$ for every $x,y\in\NB\cup\{0\}$ and {\em additive} if $f(x+y)=f(x)+f(y)$. We call a BNPG game $(\GG=(\VV,\EE),(g_v)_{v\in\VV},(c_v)_{v\in\VV})$ sub-additive (respectively additive) if $g_v$ is sub-additive (respectively additive) for every $v\in\VV$. 

%\textbf{Parameterized Complexity.}
\longversion{
In parameterized
complexity, each problem instance comes
with a parameter $k$. Formally, a parameterized problem $\Pi$ is a
subset of $\Gamma^{*}\times
\mathbb{N}$, where $\Gamma$ is a finite alphabet. An instance of a
parameterized problem is a tuple $(x,k)$, where $k$ is the
parameter. A central notion is \emph{fixed parameter
	tractability} (FPT) which means, for a
given instance $(x,k)$, solvability in time $f(k) \cdot p(|x|)$,
where $f$ is an arbitrary computable function of $k$ and
$p$ is a polynomial in the input size $|x|$.
We use the notation $\OO^*(f(k))$ to denote $O(f(k)poly(|x|))$. 
While there have been important examples of traditional algorithms that have been analyzed in this fashion, the theoretical foundations for deliberate design of such algorithms, and a complementary complexity-theoretic framework for hardness, were developed in the late nineties~\cite{CyganFKLMPPS15}.
Just as NP-hardness is used as an evidence that a problem is unlikely to be polynomial time solvable, there exists a hierarchy of complexity classes above FPT like $W[1]$, $W[2]$, and showing that a parameterized problem is hard for one of these classes is considered
evidence that the problem is unlikely to be fixed-parameter tractable. 
%Indeed, assuming the Exponential Time Hypothesis~\cite{DBLP:journals/jcss/ImpagliazzoP01}, a problem that is hard for~\WO{} does not belong to \FPT{}. The main classes in this hierarchy are:
%
%$$ \FPT  \subseteq \WO \subseteq \WT \subseteq \cdots \subseteq \WP \subseteq \XP,$$
%
%\noindent where a parameterized problem belongs to the class \ensuremath{\mathsf{XP}} if there exists an algorithm for it with running time bounded by~$|x|^{g(k)}$, where $g$ is an arbitrary computable function of $k$. For our hardness results, we will rely on parameterized reductions from known hard problems. The notion of parameterized reduction is defined as follows.
We say a parameterized problem is \PNPH if it is \NPH even for some constant values of the parameter.
}
\longversion{Due to space constraints we refer the reader to the \Cref{std:def} for the formal introduction of Parameterized Complexity and for the definitions of \xp, treewidth, tree decomposition, elimination forest and treedepth.}

\textbf{Parameterized Complexity.}
A parameterized problem is represented by the tuple $(x, k)$, where $k$ is the parameter. \emph{Fixed parameter tractability} (FPT) refers to the solvability of a given instance $(x, k)$ in time $f(k) \cdot p(|x|)$, where $p$ is a polynomial in the input size $|x|$ and $f$ is an arbitrary computable function of $k$. We use the notation $\OO^*(f(k))$ to denote $O(f(k)poly(|x|))$. There is a hierarchy of complexity classes above FPT, such as W [1], W [2], para- NP, and showing that a parameterized problem is hard for one of these complexity classes would imply that the problem may not be fixed-parameter tractable. \textbf{XP} is the class of parameterized problems that can be solved in time $n^{f(k)}$, where $k$ is the parameter, $n$ is the input size and $f$ is some computable function. 
\begin{definition}\cite{CyganFKLMPPS15}
A \textbf{tree decomposition} of a graph $G$ is a pair $\TT=(T,\{X_y\}_{t\in V(T)})$, where $T$ is a tree whose every node $t$ is assigned a vertex subset $X_t \subseteq V (G)$, called a bag, such that the following three conditions
hold: 
\begin{enumerate}
\item  $\bigcup _{t\in V(T)} X_t = V (G)$. In other words, every vertex of $G$ is in at least
one bag.
\item For every $\{u,v\} \in E(G)$, there exists a node $t$ of $T$ such that bag $X_t$
contains both $u$ and $v$.
\item For every $u \in V (G)$, the set $T_u = \{t \in V (T) : u \in X_t\}$, i.e., the set
of nodes whose corresponding bags contain $u$, induces a connected subtree
of $T$.
\end{enumerate}
\end{definition}
\begin{definition}\cite{CyganFKLMPPS15}
The \textbf{width} of tree decomposition $\TT = (T, \{X_t\}_{t\in V (T)})$ equals max$_{t\in V (T)} |X_t|-1$, that is, the maximum size of its bag minus $1$. The \textbf{treewidth} of a graph $G$,
denoted by $tw(G)$, is the minimum possible width of a tree decomposition of
$G$.
\end{definition}
\begin{definition}\cite{iwata2017power}
An \textbf{elimination forest} $T$ of a graph $G = (V, E)$ is a rooted forest on the same vertex set $V$ such that, for every edge $\{u,v\} \in E$, one of $u$ and $v$ is an ancestor of the other. The \textbf{depth} of $T$ is the maximum number of vertices on a path from a root to a leaf in $T$. The \textbf{tree-depth} $td(G)$ of a graph G is the minimum depth among all possible elimination forests.
\end{definition}

\section{Technical Contributions} 

Our main technical contributions in this paper are the hardness results. First we show that \psne is \PNPH with respect to the maximum degree of the graph as parameter [\Cref{thm:deg}].  We prove this by exhibiting a non-trivial reduction from \tsat. Next we show that \psne is \WOH parameterized by treedepth [\Cref{thm:tw}]. We prove this by exhibiting a non-trivial reduction from \GFAC. We also show an important reduction from heterogeneous game to fully homogeneous game which allows us to prove that the hardness results for maximum degree, treedepth, diameter hold even for fully homogeneous games [\Cref{thm:full-homo,thm:fully-deg}].  

We complement the hardness result for treedepth by designing a non-trivial dynamic programming based \xp algorithm parameterized by treewidth [\Cref{thm:XP}]. Our \xp algorithm also yields a fixed-parameter tractability for the combined parameter\\ ``treewidth+maximum degree''.

Lastly, using some standard techniques, we bridge the gap between tractibility and intractibility by showing (i) \WT-hardness for the parameters- the number of participating (who play $1$) and non-participating players (who play $0$) [\Cref{thm:k0,thm:k1}], (ii) fixed-parameter tractability for parameters like vertex-cover number [for strict games], circuit rank and distance from complete graphs [\Cref{vertex-cover1,thm:d1,thm:d2}] and (iii) existence of PSNE in Fully homogeneous games for important classes of graphs like  path, complete graph, cycle, and bi-clique [\Cref{thm:rest}].

\section{Results}

%Before we state our results, we would like to mention one key property about the BNPG game(\GG,\UU). Finding PSNE in a BNPG game(\GG,\UU) where \GG has $k$ connected components is equivalent to finding PSNE in each of the $k$ connected components separately. This is because the best response for a player only depends on its neighbours. So this implies that a player $u$ affect the best response of $v$ iff there is between $u$ and $v$ and therefore none of the players in one component can affect the best responses of a players other components. 

We begin with presenting our results for \psne. We omit some proofs; they are marked $\star$. They are available in the appendix.
\subsection{Hardness Results}

The \bnpg problem is already known to be \NPC~\cite{yu2020computing}. We strengthen this result significantly in \Cref{thm:deg} by proving para-\NP-hardness by the maximum degree and the number of different utility functions. We use the \NPC problem \tsat to prove some of our hardness results~\cite{berman2004approximation}. The \tsat problem is the 3-SAT problem restricted to formulas in which each clause contains exactly three literals, and each variable occurs exactly twice positively and twice negatively.

\begin{theorem} \label{thm:deg}
	\psne is \NPC for sub-additive strict BNPG games even if the underlying graph is $3$-regular and the number of different utility functions is $2$. In particular, \psne parameterized by (maximum degree $\Delta$, the number of different utility functions) is \PNPH even for sub-additive strict BNPG games.
\end{theorem}
\begin{proof}
	The \psne problem clearly belongs to \NP. To show its \NP-hardness, we reduce from the \tsat problem. The high-level idea of our proof is as follows. For every clause in \tsat instance, we create a vertex in the \psne instance. Also, for every literal we create a vertex in the \psne instance. We then add the set of edges and define the best-response functions in such a way that all the clause vertices play $1$ in any PSNE and a set of literal vertices play $1$ in a PSNE if and only if there is a satisfying assignment where the same set of literal vertices is assigned \true. We now present our construction formally.
	
	Let $(\XX=\{x_i:{i\in[n]}\}, \CC=\{C_j: j\in[m]\})$ be an arbitrary instance of \tsat. We define a function $f:\{x_i,\bar{x}_i: i\in[n]\}\longrightarrow\{a_i,\bar{a}_i: i\in[n]\}$ as $f(x_i)=a_i$ and $f(\bar{x}_i)=\bar{a}_i$ for $i\in[n]$ and consider the following instance $(\GG=(\VV,\EE), (g_v)_{v\in\VV}, (c_v)_{v\in\VV})$ of \psne.	
	\begin{align*}
	\VV &= \{a_i, \bar{a}_i: i\in[n]\} \cup \{y_j: j\in[m]\}\\
	\EE &= \{\{y_j,f(l_1^j)\}, \{y_j,f(l_2^j)\}, \{y_j,f(l_3^j)\}:
	C_j = (l_1^j\vee l_2^j\vee l_3^j),\\& j\in[m]\}
	\cup \{\{a_i,\bar{a}_i\}: i\in[n]\}
	\end{align*}
	
	We observe that the degree of every vertex in \GG is $3$. We now define $(g_v)_{v\in\VV}$ and $(c_v)_{v\in\VV}$. $\forall j\in[m],\text{ }c_{y_j}=4,g_{y_j}(0)=1000,g_{y_j}(1)=1003,g_{y_j}(2)=1008,g_{y_j}(3)=1013,g_{y_j}(4)=1018.$ $\forall i\in[n],\text{ }c_{a_i}=c_{\bar{a}_i }=4, g_{a_i}(0)=g_{\bar{a}_i}(0)= 1000, g_{a_i}(1)=g_{\bar{a}_i}(1)=1005,g_{a_i}(2)=g_{\bar{a}_i}(2)=1010,g_{a_i}(3)=g_{\bar{a}_i}(3)=1015,g_{a_i}(4)=g_{\bar{a}_i}(4)=1018.$ 
	
	It follows from the definition that both the above functions are sub-additive. Also, one can easily verify that the above functions give the following best-response functions for the players.	
	\[\forall i\in[n], \beta_{a_i}(k) = \beta_{\bar{a}_i}(k) = \begin{cases}
	1 & \text{if } k\le 2\\
	0 & \text{otherwise}
	\end{cases}
	\]
	\[
	\forall j\in[m], \beta_{y_j}(k) = \begin{cases}
	0 & \text{if }k=0\\
	1 & \text{otherwise}
	\end{cases}
	\]
	From the best-response functions, it follows that the game is strict. We now claim that the above BNPG game has a PSNE if and only if the \tsat instance is a \YES instance.

	For the ``if'' part, suppose the \tsat instance is a \YES instance. Let $h:\{x_i: i\in[n]\}\longrightarrow\{\true, \false\}$ be a satisfying assignment of the \tsat instance. We consider the following strategy profile for the BNPG game.
	\begin{itemize}
		\item $\forall j\in[m],s(y_j)=1$
		\item $\forall i\in[n],s(a_i)=1$ if and only if $h(x_i)=\true$
		\item $\forall i\in[n],s(\bar{a}_i)=0$ if and only if $h(x_i)=\true$
	\end{itemize}
	We observe that, since $h$ is a satisfying assignment, the player $y_j$ for every $j\in[m]$ has at least one neighbor who plays $1$ and thus $y_j$ does not have any incentive to deviate (from playing $1$). For $i\in[n]$ such that $h(x_i)=\true$, the player $a_i$ has at least one neighbor, namely $\bar{a}_i$, who plays $0$ and thus $a_i$ does not have any incentive to deviate (from playing $1$); on the other hand the player $\bar{a}_i$ has all her neighbor playing $1$ , and thus she is happy to play $0$. Similarly, for $i\in[n]$ such that $h(x_i)=\false$, both the players $a_i$ and $\bar{a}_i$ have no incentive to deviate. This proves that the above strategy profile is a PSNE.
	
	For the ``only if'' part, let $(s(a_i)_{i\in[n]}, s(\bar{a}_i)_{i\in[n]},\allowbreak s(y_j)_{j\in[m]})$ be a PSNE for the BNPG game. We claim that $s(y_j)=1$ for every $j\in[m]$. Suppose not, then there exists a $t\in[m]$ such that $s(y_t)=0$. Let the literals in clause $C_t$ be $l_1^t,l_2^t,l_3^t$. Then $s(f(l_i^t))=0,\forall i\in[3]$ otherwise the player $y_t$ will deviate form $0$ and play $1$. But then the player $f(l_1^t)$ will deviate to $1$ as $y_t$ plays $0$ which is a contradiction. We now claim that we have $s(a_i) \ne s(\bar{a}_i)$ for every $i\in[n]$. Suppose not, then there exists an $\lambda\in[n]$ such that $s(a_\lambda) = s(\bar{a}_\lambda)$. If $s(a_\lambda) = s(\bar{a}_\lambda)=1$, then both the players $a_\lambda$ and $\bar{a}_\lambda$ have incentive to deviate to $0$. On the other hand, if $s(a_\lambda) = s(\bar{a}_\lambda)=0$, then both the players $a_\lambda$ and $\bar{a}_\lambda$ have incentive to deviate to $1$. This proves the claim. We now consider the assignment  $h:\{x_i: i\in[n]\}\longrightarrow\{\true, \false\}$ defined as $h(x_i)=\true$ if and only if $s(a_i)=1$ for every $i\in[n]$. We claim that $h$ is a satisfying assignment for the \tsat formula. Suppose not, then $h$ does not satisfy a clause, say $C_\gamma, \gamma\in[m]$. Then the player $y_\gamma$ has incentive to deviate to $0$ as none of its neighbors play $1$ which is a contradiction.
\end{proof}

For the remainder of this subsection, we describe a game using the best response functions for the sake of simplicity of presentation. This suffices as due to \Cref{lem:best-util}, we can always compute the utility functions using the best response functions in polynomial time.

We next consider treedepth as parameter. Problems on graphs which are easy for trees are often fixed-parameter-tractable with respect to treedepth as parameter. We show that this is not the case for our problem.
Towards that, we use the \GFAC problem which is \WOH parameterized by treedepth~\cite{samer2011tractable}.

\begin{definition}[\GFAC]
	Given a graph $\GG=(\VV,\EE)$ and a set $K(v)\subseteq\{0,...,d(v)\}$ for each $v\in\VV$, compute if there exists a subset $\FF\subseteq \EE$ such that, for each vertex $v\in \VV$, the number of edges in $\FF$ incident on $v$ is an element of $K(v)$. We denote an arbitrary instance of this problem by $(\GG=(\VV,\EE),(K(v))_{v\in\VV})$.
\end{definition}

\begin{figure}
	\centering
\begin{tikzpicture}
	\draw (0,0) ellipse (3cm and .8 cm);
	\draw (-2,0) circle [radius=0.3] node (u1) {$u_1$};
	\draw (-1,0) circle [radius=0.3] node (u2) {$u_2$};
	\draw (0,0) node {$\cdots$};	
	\draw (2,0) circle [radius=0.3] node (un) {$u_n$};	
	
	\draw (-2,-1.3) circle [radius=0.3] node (u1p) {$u_1^\pr$};
	\draw (-1,-1.3) circle [radius=0.3] node (u2p) {$u_2^\pr$};
	\draw (0,-1.3) node {$\cdots$};	
	\draw (2,-1.3) circle [radius=0.3] node (unp) {$u_n^\pr$};
	\draw (4,-1.3) circle [radius=0.4] node (un1p) {$u_{n+1}^\pr$};
	
	\draw (0,-2.4) circle [radius=0.3] node (d1) {$d_1$};
	\draw (4,1.5) circle [radius=0.3] node (d2) {$d_2$};
	
	\draw (-1.8,1.5) circle [radius=0.5] node (a12) {$a_{\{1,2\}}$};
	\draw (-.8,1.5) node {$\cdots$};
	
	\draw[dashed] (u1) -- (u2);
	\draw[-] (u1) -- (a12);
	\draw[-] (u2) -- (a12);
	
	\draw[-] (u1) -- (u1p);
	\draw[-] (u2) -- (u2p);
	\draw[-] (un) -- (unp);	
	
	\draw[-] (d1) -- (u1p);
	\draw[-] (d1) -- (u2p);
	\draw[-] (d1) -- (unp);
	\draw[-] (d1) -- (un1p.south);
	
	\draw[-] (d2) -- (u1.north);
	\draw[-] (d2) -- (u2);
	\draw[-] (d2) -- (un);
	\draw[-] (d2) -- (un1p.north);
\end{tikzpicture}
\caption{Graph \HH in the proof of \Cref{thm:tw}.}\label{fig:tw}
\end{figure}

\begin{theorem}\label{thm:tw}
	\psne for BNPG games is \WOH parameterized by treedepth.
\end{theorem}

\begin{proof}
	To prove \WO-hardness, we reduce from \GFAC parameterized by treedepth to BNPG game. 

Let $\left(\GG=\left(\left\{v_i: i\in[n]\right\},\EE^\pr\right),\left(K\left(v_i\right)\right)_{i\in[n]}\right)$ be an arbitrary instance of \GFAC. The high level idea of our construction is as follows. For each vertex and edge in the graph $\GG$ associated with \GFAC instance, we add a node in the graph $\HH$ (where the BNPG game is defined) associated with \psne problem instance. On top of that we add some extra nodes and edges in \HH and appropriately define the best response functions of every player in \HH so that a set of nodes in \HH corresponding to a set \FF of edges belonging to \GG play $1$ in a PSNE if and only if \FF makes \GFAC a yes instance. We now formally present our construction.
	
	We consider a BNPG game on the following graph $\HH=(\VV,\EE)$. See \Cref{fig:tw} for a pictorial representation of \HH.	
	\begin{align*}
		\VV &= \{u_i:i\in[n]\}\cup\{a_{\{i,j\}}: \{v_i,v_j\}\in\EE^\pr\}\\
		&\cup \{u_i^\pr: i\in[n+1]\} \cup \{d_1,d_2\}\\
		\EE &= \{\{u_i,a_{\{i,j\}}\}, \{u_j,a_{\{i,j\}}\}: \{v_i,v_j\}\in\EE^\pr\}
		\cup \{\{u_i,u_i^\pr\}:i\in[n]\}\\
		&\cup \{\{d_1,u_i^\pr\}, \{d_2,u_i\}:i\in[n]\}
		 \cup\{\{d_1,u_{n+1}^\pr\},\{d_2,u_{n+1}^\pr\}\}
	\end{align*}

	Let the treedepth of \GG be $\tau$. Create a graph $\GG^\pr$ by adding the vertices $d_1,d_2$ and the set of edges $\{\{d_1,v_i\},\{d_2,v_i\}:i\in[n]\}\cup\{d_1,d_2\}$ to the graph \GG. The treedepth of $\GG^\pr$ is at most $\tau+2$. We claim that the treedepth of \HH is at most $\tau+3$. To see this, we begin with a elimination tree of $\GG^\pr$ and replace $v_i$ with $u_i$ for every $i\in[n]$. Let $\SS=\{u_i^\pr:i\in[n+1]\}\cup\{a_{\{i,j\}}: \{v_i,v_j\}\in\EE^\pr\}$. $\forall u^\pr\in\SS$, add an edge between $u^\pr$ and $u$ in the elimination tree where $u$,$v$ are neighbors of $u^\pr$ in \HH and $u$ is descendant of $v$ in the elimination tree. This results in a valid elimination tree for \HH and hence, the treedepth of \HH is at most $\tau+3$.
	
	We now describe the best-response functions of the vertices in \HH to complete the description of the BNPG game.
	\[
	\forall i\in[n], \beta_{u_i}(k) = \begin{cases}
	1 & \text{if } k-1\in K(v_i)\\
	0 & \text{otherwise }
	\end{cases}
	\]
	\[
 \forall i\in[n+1], \beta_{u_i^\pr}(k) = \begin{cases}
	1 & \text{if } k=2\\
	0 & \text{otherwise }
	\end{cases}
	\]
	\[
	\forall \{v_i,v_j\}\in\EE^\pr, \beta_{a_{\{i,j\}}}(k) = \{0,1\}\;\forall k\in\NB\cup\{0\}
	\]
	\[ \beta_{d_1}(k)=\begin{cases}
	1 & \text{if } k=0 \text{ or } k=n\\
	0 & \text{otherwise }
	\end{cases},
	\beta_{d_2}(k)=\begin{cases}
	1 & \text{if } k=0\\
	0 & \text{otherwise }
	\end{cases}
	\]
	
	We claim that the above BNPG game has a PSNE if and only if the \GFAC instance is a \yes instance. %\shortversion{ Let $\bar{x}=(x_v)_{v\in\VV}$ be a PSNE of the BNPG game. We can show that in any PSNE, $x_{d_1}=1, x_{u_i}=x_{u_i^\pr}=1, \forall i\in[n], x_{u_{n+1}^\pr}=0, x_{d_2}=0$. Due to this property it turns out that $\bar{x}=(x_v)_{v\in\VV}$ is a PSNE if and only if the set $\{\{v_i,v_j\}:x_{a_{\{i,j\}}}=1,\{v_i,v_j\}\in \EE^\pr\}$ is a feasible solution of the \GFAC instance. Due to space constraints, we refer the reader to the supplementary material for the detailed proof of this theorem.}

	For the ``if'' part, suppose the \GFAC instance is a \yes instance. Then there exists a subset $\FF\subseteq\EE^\pr$ such that for all $i\in[n]$, the degree of $v_i$ in $\GG[\FF]$ is an element of the set $K(v_i)$. We consider the strategy profile $\bar{x}=(x_v)_{v\in\VV}$.	
	\[ \forall i\in[n], x_{u_i}=x_{u_i^\pr}=1, x_{u_{n+1}^\pr}=0\]
	\[
	 \forall \{v_i,v_j\}\in\EE^\pr, x_{a_{\{i,j\}}}=\begin{cases}
	1 & \text{if } \{v_i,v_j\}\in\FF\\
	0 & \text{otherwise }
	\end{cases} ,x_{d_1}=1, x_{d_2}=0 \]
	
	Now we argue that $\bar{x}$ is a PSNE for the BNPG game. Clearly no player $a_{\{i,j\}}, \{v_i,v_j\}\in\EE^\pr$ deviates as both $0$ and $1$ are her best-responses irrespective of the action of their neighbors. The player $d_1$ does not deviate as she has exactly $n$ neighbors playing $1$. The player $u_i', i\in[n]$ does not deviate as she has exactly $2$ neighbors playing $1$. The player $u_{n+1}^\pr$ does not deviate as she has exactly $1$ neighbor playing $1$. The player $d_2$ does not deviate as she has at least $1$ neighbors playing $1$. Note that $\forall i\in[n]$, the number of neighbors of $u_i$ playing $1$ excluding $u_i^\pr$ and $d_2$ (which in this case is $n_{u_i}-1$ as $x_{d_2}=0, x_{u_i^\pr}=1$) is the same as the number of edges in $\FF$ which are incident on $v_i$ in \GG. Hence, $\forall i \in [n]$, the player $u_i$ does not deviate as $(n_{u_i}-1)\in K(v_i)$. Hence, $\bar{x}$ is a PSNE.
	
	For the ``only if'' part, let $\bar{x}=(x_v)_{v\in\VV}$ be a PSNE of the BNPG game. We claim that we have $x_{d_1}=1, x_{u_i}=x_{u_i^\pr}=1, \forall i\in[n], x_{u_{n+1}^\pr}=0, x_{d_2}=0$. To prove this, we consider all cases for $(x_{u_i})_{i\in[n]}$.
	\begin{enumerate}
		\item Case -- $\forall i \in [n]$ $x_{u_i}=1$: We have $x_{d_2}=0$ as $n_{d_2}>0$ otherwise $d_2$ would deviate. This implies that $x_{u_{n+1}^\pr}=0$ since $n_{u_{n+1}^\pr}\leq 1$( as $x_{d_2}=0$). Now we consider the following sub-cases (according to the values of $x_{d_1}$ and $x_{u_i^\pr}, i\in[n]$):
		\begin{itemize}
			\item ($x_{d_1}=1, \exists k\in [n]$ such that $x_{u_{k}^\pr }=0$ ). Here $x_{u_k^\pr}$ will then deviate to $1$ as $n_{u_k^\pr}=2$. Hence, it is not a PSNE.
			
			\item ($x_{d_1}=1, \forall i\in[n]$ $x_{u_i^\pr}=1$). This is exactly what we claim thus we have nothing to prove in this case.
			
			\item ($x_{d_1}=0, \exists k\in [n]$ such that $x_{u_{k}^\pr }=1$). Here $x_{u_k^\pr}$ will then deviate to $0$ as $n_{u_k^\pr}=1$. Hence, it is not a PSNE.
			
			\item ($x_{d_1}=0$,$\forall i\in[n]$ $x_{u_i^\pr}=0$). The player $d_1$ will deviate to $1$ as $n_{d_1}=0$. Hence, it is not a PSNE.
		\end{itemize}
	
		\item Case -- $\exists k_1,k_2\in [n]$ such that $x_{u_{k_1}}=1$ and $x_{u_{k_2}}=0$: We have $x_{d_2}=0$ as $n_{d_2}>0$ otherwise $d_2$ would deviate. This implies that $x_{u_{n+1}^\pr}=0$ since $n_{u_{n+1}^\pr}\leq 1$ (as $x_{d_2}=0$). Now we consider the following sub-cases (according to the values of $x_{d_1}$ and $x_{u_i^\pr}, i\in[n]$):
		\begin{itemize}
			
			\item ($x_{d_1}=1$,$\forall i\in[n]$ $x_{u_i^\pr}=0$). Here $u_{k_1}^\pr$ will deviate to 1 as $n_{u_{k_1}^\pr}=2$. So, it isn't a PSNE.
			
			\item ($x_{d_1}=1$,$\forall i\in[n]$ $x_{u_i^\pr}=1$). Here $u_{k_2}^\pr$ will deviate to 0 as $n_{u_{k_2}^\pr}=1$. So, it isn't a PSNE.
			
			\item ($x_{d_1}=1$, $\exists i,j\in [n]$ such that $x_{u_{i}^\pr}=1$  and $x_{u_{j}^\pr}=0$ ). Here $d_1$ will deviate to $0$ as $0<n_{d_1}<n$ (there are at least $2$ neighbours of $d_1$ which play $0$ and at least $1$ neighbour of $d_1$ which plays $1$). Hence, it is not a PSNE.
			
			\item ($x_{d_1}=0$,$\forall i\in[n]$ $x_{u_i^\pr}=0$). Here $d_1$ will deviate to 1 as $n_{d_1}=0$. So, it isn't a PSNE.
			
			\item ($x_{d_1}=0$,$\exists i\in [n]$ such that $x_{u_{i}^\pr}=1$). Here $u_{i}^\pr$ will deviate to 0 as $n_{u_{i}^\pr}\leq1$ and hence, it is not a PSNE.
		\end{itemize}
	
		\item Case -- $\forall i \in [n]$ $x_{u_i}=0$: For every $i\in[n]$, we must have $x_{u_i^\pr}=0$ so that $u_i^\pr$ doesn't deviate. We have the following sub-cases (according to the values of $x_{d_1},x_{d_2}$ and $x_{u_{n+1}^\pr}$):
		\begin{itemize}
			\item $(x_{d_1}=0,x_{u_{n+1}^\pr}=0).$ Here $d_1$ deviates to 1 as $n_{d_1}=0$ and hence, it is not a PSNE.
			\item $(x_{d_1}=0,x_{u_{n+1}^\pr}=1).$ Here $u_{n+1}^\pr$ deviates to 0 as $n_{u_{n+1}^\pr}\leq1$. So, it isn't a PSNE.
			\item $(x_{d_1}=1,x_{u_{n+1}^\pr}=0,x_{d_2}=0).$ Here $d_2$ deviates to 1 as $n_{d_2}=0$. So, it isn't a PSNE.
			\item $(x_{d_1}=1,x_{u_{n+1}^\pr}=0,x_{d_2}=1).$ Here $u_{n+1}^\pr$ deviates to 1 as $n_{u_{n+1}^\pr}=2$ and hence, it is not a PSNE.
			\item $(x_{d_1}=1,x_{u_{n+1}^\pr}=1,x_{d_2}=0).$Here $u_{n+1}^\pr$ deviates to 0 as $n_{u_{n+1}^\pr}=1$ and hence, it is not a PSNE.
			\item $(x_{d_1}=1,x_{u_{n+1}^\pr}=1,x_{d_2}=1).$Here $d_2$ deviates to 0 as $n_{d_2}>0$. So, it isn't a PSNE.
		\end{itemize}
	\end{enumerate}
	
	So if $\bar{x}=(x_v)_{v\in\VV}$ is a PSNE of the BNPG game, then we have $x_{d_1}=1$, $\forall i\in[n], x_{u_i^\pr}=1, x_{u_{n+1}^\pr}=0, \forall i \in [n]$ $x_{u_i}=1, x_{d_2}=0$. Now consider the set $\FF=\{\{v_i,v_j\}:x_{a_{\{i,j\}}}=1,\{v_i,v_j\}\in \EE^\pr\}$. Note that $\forall i\in[n]$, the number of neighbors of $u_i$ playing $1$ excluding $u_i^\pr$ and $d_2$ (which in this case is $n_{u_i}-1$ as $x_{d_2}=0, x_{u_i^\pr}=1$) is the same as the number of edges in $\FF$ which are incident on $v_i$ in \GG. Since  $\forall i, n_{u_i}-1\in K(v_i)$, the number of edges in $\FF$ incident on $v_i$ in \GFAC instance is an element of $K(v_i)$. Hence, the \GFAC instance is a \yes instance.
\end{proof}

\begin{corollary}
\psne for BNPG games is \WOH parameterized by treewidth and pathwidth.
\end{corollary}

\longversion{
\begin{proof}
This follows from the fact that treewidth and pathwidth are upper bounded by treedepth \cite{Nesetril2012}.
\end{proof}
}
%
%\begin{figure}[h]
%	
%	\centering
%	\includegraphics[width=0.5\textwidth]{9}
%	\caption{General Factor problem instance $G(V,E)$}
%\end{figure}
%\begin{figure}[h]
%	
%	\centering
%	\includegraphics[width=0.5\textwidth]{10}
%	\caption{Reduction to BNPG game. Black vertices represent a node $u_i$, yellow vertices represents a node $a_{\{i,j\}}$, blue vertices represents a node $u_i^\pr$, red node is $d_1$ and green node is $d_2$}
%\end{figure}

We next consider the diameter ($d$) of the graph as our parameter and prove para-\NP-hardness in \Cref{thm:dia}. It follows immediately from the fact that the reduced instance in the \NP-completeness proof of \psne for BNPG games in \cite{yu2020computing} has diameter $2$.

\begin{observation}\label{thm:dia}
	\psne for BNPG games is \NPC even for graphs of diameter at most $2$. In particular, the \psne problem for BNPG games is \PNPH parameterized by diameter.
\end{observation}

We next consider a variant of \psne where at most $k_0$ (respectively $k_1$) players are  playing $0$ (respectively $1$) in the PSNE. We denote this variant as \text{$k_0$-}\psne (resp.  \text{$k_1$-}\psne). Obviously there is a brute force \xp algorithm which runs in time $\OO^*\left(n^{k_0}\right)$ (respectively $\OO^*\left(n^{k_1}\right)$). We show that  \text{$k_0$-}\psne (resp.  \text{$k_1$-}\psne) is \WTH parameterized by $k_0$ (respectively $k_1$). For this, we reduce from the \ds problem parameterized by the size of dominating set which is known to be \WTH~\cite{CyganFKLMPPS15}.

\begin{theorem}[$\star$]\label{thm:k0}
	  \text{$k_0$-}\psne for BNPG games is \WTH parameterized by $k_0$.
\end{theorem}

\longversion{
\begin{proof}
Let $d(v)$ denote the degree of $v$ in \GG. To prove the result for the parameter $k_0$, we use the following best-response function.
	\[ \beta_v(k^\pr) = \begin{cases}
	0 & \text{if } k^\pr=d(v)\\
	\{0,1\} & \text{otherwise}
	\end{cases} \]

	We claim that the above BNPG game has a PSNE having at most $k$ players playing $0$ if and only if the \ds instance is a \yes instance.
	
	For the ``if'' part, suppose the \ds instance is a \yes instance and $\WW\subseteq\VV$ be a dominating set for \GG of size at most $k$. We claim that the strategy profile $\bar{x}=((x_v=0)_{v\in\WW}, (x_v=1)_{v\in\VV\setminus\WW})$ is a PSNE for the BNPG game. To see this, we observe that every player $w \in \VV\setminus\WW$ has at least $1$ neighbor playing $0$ and thus she has no incentive to deviate as $n_w<d(v)$. On the other hand, since $0$ is always a best-response strategy for every player irrespective of what others play, the players in \WW also do not have any incentive to deviate. Hence, $\bar{x}$ is a PSNE.
	
	For the ``only if'' part, let $\bar{x}=((x_v=0)_{v\in\WW}, (x_v=1)_{v\in\VV\setminus\WW})$ be a PSNE for the BNPG game where $|\WW|\le k$ (that is, at most $k$ players are playing $0$). We claim that \WW forms a dominating set for \GG. Indeed, this claim has to be correct, otherwise there exists a vertex $w\in\VV\setminus\WW$ which does not have any neighbor in \WW and consequently, the player $w$ has incentive to deviate to $0$ from $1$ as $n_w=d(v)$ which is a contradiction.
\end{proof}
}

\begin{theorem}[$\star$]\label{thm:k1}
	  \text{$k_1$-}\psne for BNPG games is \WTH parameterized by $k_1$ even for fully homogeneous BNPG games.
\end{theorem}

\longversion{
\begin{proof}
	Let $(\GG=(\VV,\EE),k)$ an arbitrary instance of \ds. We consider a fully homogeneous BNPG game on the same graph \GG whose best-response functions $\beta_v(\cdot)$ for $v\in\VV$ is given below:
	\[ \beta_v(k^\pr) = \begin{cases}
	1 & \text{if } k^\pr=0\\
	\{0,1\} & \text{otherwise}
	\end{cases} \]
	
	We claim that the above BNPG game has a PSNE having at most $k$ players playing $1$ if and only if the \ds instance is a \yes instance.
	
	For the ``if'' part, suppose the \ds instance is a \yes instance and $\WW\subseteq\VV$ be a dominating set for \GG of size at most $k$. We claim that the strategy profile $\bar{x}=((x_v=1)_{v\in\WW}, (x_v=0)_{v\in\VV\setminus\WW})$ is a PSNE for the BNPG game. To see this, we observe that every player in $\VV\setminus\WW$ has at least $1$ neighbor playing $1$ , and thus she has no incentive to deviate. On the other hand, since $1$ is always a best-response strategy for every player irrespective of what others play, the players in \WW also do not have any incentive to deviate. Hence, $\bar{x}$ is a PSNE.
	
	For the ``only if'' part, let $\bar{x}=((x_v=1)_{v\in\WW}, (x_v=0)_{v\in\VV\setminus\WW})$ be a PSNE for the BNPG game where $|\WW|\le k$ (that is, at most $k$ players are playing $1$). We claim that \WW forms a dominating set for \GG. Indeed, this claim has to be correct, otherwise there exists a vertex $w\in\VV\setminus\WW$ which does not have any neighbor in \WW and consequently, the player $w$ has incentive to deviate to $1$ from $0$ as $n_w=0$ which is a contradiction.
\end{proof}
}

Till now we have mostly focused on heterogeneous BNPG games. We next consider fully homogeneous BNPG games and show the following by reducing from the \psne problem on heterogeneous BNPG games.

\begin{theorem}\label{thm:full-homo}
	The following results hold even for fully homogeneous games.
	\begin{enumerate}
		\item \psne is \NPC even if the diameter of the graph is at most $4$.
		\item \psne is \WOH with respect to the parameter treedepth of the graph.
		\item  \text{$k_0$-}\psne is \WTH parameterized by $k_0$.
	\end{enumerate} 
\end{theorem}
\begin{proof}%[Proof of \Cref{thm:full-homo}]
	We first present a reduction from the \psne problem on heterogeneous BNPG games to the \psne problem on fully homogeneous BNPG games. Let $(\GG=(\VV=\{v_i: i\in[n]\},\EE),(g_v)_{v\in\VV},(c_v)_{v\in\VV})$ be any heterogeneous BNPG game. We now construct the graph $\HH=(\VV^\pr,\EE^\pr)$ for the instance of the fully homogeneous BNPG game. 
	\begin{align*}
		\VV^\pr &= \{u_i: i\in[n]\} \cup\bigcup_{i\in[n]} \VV_i, \\
		&\text{where }
		\VV_i = \{a_j^i: j\in[2+(i-1)n]\}, \forall i\in[n]\\
		\EE^\pr &= \{\{u_i,u_j\}: \{v_i,v_j\}\in\EE\}\cup\bigcup_{i\in[n]} \EE_i,\\
		& \text{where }
		\EE_i = \{\{a_j^i,u_i\}:j\in[2+(i-1)n]\}, \forall i\in[n]
	\end{align*}
	
	Let us define $f(x)=\lfloor\frac{x-2}{n}\rfloor+1, h(x)=x-2-(f(x)-1)n$. We now define best-response strategies $\beta$ for the fully homogeneous BNPG game on \HH.
	\[
	\beta(k)=
	\begin{cases}
		1 & \text{if } k=0 \text{ or } k=1\\
		\{0,1\} & \text{if } \Delta g_{v_{f(k)}}(h(k))=c_{v_{f(k)}},k>1 \\
		1 & \text{if } \Delta g_{v_{f(k)}}(h(k))>c_{v_{f(k)}},k>1 \\
		0 & \text{if } \Delta g_{v_{f(k)}}(h(k))<c_{v_{f(k)}},k>1
	\end{cases}
	\]
	
	This finishes description of our fully homogeneous BNPG game on \HH. We now claim that there exists a PSNE in the heterogeneous BNPG game on \GG if and only if there exists a PSNE in the fully homogeneous BNPG game on \HH.
	
	For the ``only if'' part, let $x^*=(x^*_v)_{v\in \VV}$ be a PSNE in the heterogeneous BNPG game on \GG. We now consider the following strategy profile $\bar{y}=(y_v)_{v\in\VV^\pr}$ for players in \HH.
	\[
	\forall i\in[n] y_{u_i} = x^*_{v_i}; y_w=1 \text{ for other vertices }w
	\]
	
	Clearly the players in $\cup_{i\in[n]} \VV_i$ do not deviate as their degree is $1$ and $\beta(0)=\beta(1)=1$. In $\bar{y}$, we have $n_{u_i} = n_{v_i}+2+(i-1)n\ge 2$ and $n_{v_i}\le n-1$ for every $i\in[n]$. If $x_{v_i}^*=1$, then we have $\Delta g_{v_i}(n_{v_i})\geq c_{v_i}$. We have $f(n_{u_i})=i$ and $h(n_{u_i})=n_{v_i}$. This implies that $\Delta g_{v_{f(n_{u_i})}}(h(n_{u_i}))\geq c_{v_{f(n_{u_i})}}$. So $u_i$ does not deviate as $1$ is the best-response. If $x_{v_i}^*=0$, then we have $\Delta g_{v_i}(n_{v_i})\leq c_{v_i}$. This implies that $\Delta g_{v_{f(n_{u_i})}}(h(n_{u_i}))\leq c_{v_{f(n_{u_i})}}$. So $u_i$ does not deviate as $0$ is the best-response. Hence, $\bar{y}$ is a PSNE.
	
	For the ``if'' part, suppose there exists a PSNE $(x^*_v)_{v\in \VV^\pr}$ in the fully homogeneous BNPG game on \HH.  Clearly $x^*_v=1$ for all $v\in \cup_{i\in[n]} \VV_i$ as $n_v\leq 1$. Now we claim that the strategy profile $\bar{x}=(x_{v_i}=x_{u_i}^*)_{i\in[n]}$ forms a PSNE for the heterogeneous BNPG game on \GG. We observe that if $x_{u_i}^*=1$, then  $\Delta g_{v_{f(n_{u_i})}}(h(n_{u_i}))\geq c_{v_{f(n_{u_i})}}$ for $i\in[n]$. This implies that $\Delta g_{v_i}(n_{v_i})\geq c_{v_i}$.  So $x_{v_i}=1$ is the best-response for $v_i$ in \GG and hence, she does not deviate. Similarly, If $x_{u_i}^*=0$, then $\Delta g_{v_{f(n_{u_i})}}(h(n_{u_i}))\leq c_{v_{f(n_{u_i})}}$. This implies that $\Delta g_{v_i}(n_{v_i})\leq c_{v_i}$.  So $x_{v_i}=0$ is the best-response for $v_i\in \VV$ and hence, it won't deviate. Hence, $\bar{x}$ is a PSNE in the heterogeneous BNPG game on \GG.
	
	We now prove the three statements in the theorem as follows.
	\begin{enumerate}
		\item We observe that, if the diameter of \GG is at most $2$, then the diameter of \HH is at most $4$. Hence, the result follows from \Cref{thm:dia}.
		
		\item We observe that the treedepth of \HH is at most $1$ more than the treedepth of \GG. Hence, the result follows from \Cref{thm:tw}.
		\item We observe that there exists a PSNE where at most $k$ players play $0$ in the heterogeneous BNPG game on \GG if and only if there exists a PSNE where at most $k$ players play $0$ in the fully homogeneous BNPG game on \HH. Hence, the result follows from \Cref{thm:k0}.
	\end{enumerate}
\end{proof}

We next show that \psne for fully homogeneous BNPG games is \PNPH parameterized by the maximum degree of the graph again by reducing from heterogeneous BNPG games.

\begin{theorem}[$\star$]\label{thm:fully-deg}
	\psne for fully homogeneous BNPG games is \NPC even if the maximum degree $\Delta$ of the graph is at most $9$.
\end{theorem}

\longversion{
\begin{proof}
	The high-level idea in this proof is the same as the proof of \Cref{thm:full-homo}. The only difference is that we add the special nodes in a way such that the maximum degree in the instance of fully homogeneous BNPG game is upper bounded by $9$.
	
	Formally, we consider an instance of a heterogeneous BNPG game on a graph $\GG=(\VV=\{v_i: i\in[n]\},\EE)$ such that there are only $2$ types of utility functions $U_1(x_v,n_v)=g_1(x_v+n_v)-c_1x_v, U_2(x_v,n_v)=g_2(x_v+n_v)-c_2x_v,$ and the degree of any vertex is at most $3$; we know from \Cref{thm:deg} that it is an \NPC instance. Let us partition $\VV$ into $\VV_1$ and $\VV_2$ such that the utility function of the players in $\VV_1$ is $U_1(\cdot)$ and the utility function of the players in $\VV_2$ is $U_2(\cdot)$. We now construct the graph $\HH=(\VV^\pr,\EE^\pr)$ for the instance of the fully homogeneous BNPG game. 
	\begin{align*}
		\VV^\pr &= \{w_i: i\in[n]\}\cup\WW_1\cup\WW_2, \text{where}\\
		\WW_i &= \{a_j^k: j\in[2+4(i-1)], v_k\in\VV_i\} \;\forall i\in[2]\\
		\EE^\pr &= \{\{w_i,w_j\}: \{v_i,v_j\}\in\EE\}\cup\EE_1\cup\EE_2, \text{where}\\
		\EE_i &= \{\{a_j^k,w_k\}: j\in[2+4(i-1)], v_k\in\VV_i\} \;\forall i\in[2]
	\end{align*}
	
	We define two functions --- $f(x)=\lfloor\frac{x-2}{4}\rfloor+1$ and $h(x)=x-2-4(f(x)-1)$. We now define best-response functions for the players in \HH.
	\[
	\beta(k)=\begin{cases}
	1 & \text{if } k=0 \text{ or } k=1\\
	\{0,1\} & \Delta g_{f(k)}(h(k))=c_{f(k)},k>1 \\
	1 & \Delta g_{f(k)}(h(k))>c_{f(k)},k>1 \\
	0 & \Delta g_{f(k)}(h(k))<c_{f(k)},k>1
	\end{cases}
	\]
	
	This finishes description of our fully homogeneous BNPG game on \HH. We now claim that there exists a PSNE in the heterogeneous BNPG game on \GG if and only if there exists a PSNE in the fully homogeneous BNPG game on \HH. We note that degree of any node in \HH is at most $9$.
	
	For the ``only if'' part, let $x^*=(x^*_v)_{v\in \VV}$ be a PSNE in the heterogeneous BNPG game on \GG. We now consider the following strategy profile $\bar{y}=(y_v)_{v\in\VV^\pr}$ for players in \HH.
	\[
	\forall i\in[n] y_{w_i} = x^*_{v_i}; y_b=1 \text{ for other vertices }b
	\]
	
	We now claim that the players in $\VV'$ also does not deviate. Clearly the players in $\cup_{i\in[2]} \WW_i$ do not deviate as their degree is $1$ and $\beta(0)=\beta(1)=1$. If $v_k\in\VV_i$, then in $\bar{y}$ we have $n_{w_k} = n_{v_k}+2+4(i-1)\ge 2$ for every $i\in[2]$ and $n_{v_k}\le 3$ as the maximum degree in \GG is $3$. If $x_{v_k}^*=1$, then we have $\Delta g_{i}(n_{v_k})\geq c_{i}$. We have $f(n_{w_k})=i$ and $h(n_{w_k})=n_{v_k}$. This implies that $\Delta g_{f(n_{w_k})}(h(n_{w_k}))\geq c_{f(n_{w_k})}$. So $w_k$ does not deviate as $1$ is the best-response. If $x_{v_k}^*=0$, then we have $\Delta g_{i}(n_{v_k})\leq c_{i}$. This implies that $\Delta g_{f(n_{w_k})}(h(n_{w_k}))\leq c_{f(n_{w_k})}$. So $w_k$ does not deviate as $0$ is the best-response. Hence, $\bar{y}$ is a PSNE.
	
	For the ``if'' part, suppose there exists a PSNE $(x^*_v)_{v\in \VV^\pr}$ in the fully homogeneous BNPG game on \HH. Clearly $x^*_w=1$ for all $w\in \cup_{i\in[2]} \WW_i$ as $n_w\leq 1$. Now we claim that the strategy profile $\bar{x}=(x_{v_k}=x_{w_k}^*)_{k\in[n]}$ forms a PSNE for the heterogeneous BNPG game on \GG. We observe that if $x_{w_k}^*=1$ and $v_k\in\VV_i$, then $\Delta g_{f(n_{w_k})}(h(n_{w_k}))\geq c_{f(n_{w_k})}$ for $k\in[n]$. This implies that $\Delta g_{i}(n_{v_k})\geq c_{i}$.  So $x_{v_k}=1$ is the best-response for $v_k\in\VV_i$ and hence, she does not deviate. Similarly, If $x_{w_k}^*=0$ and $v_k\in\VV_i$, then $\Delta g_{f(n_{w_k})}(h(n_{w_k}))\leq c_{f(n_{w_k})}$. This implies that $\Delta g_{i}(n_{v_k})\leq c_{i}$.  So $x_{v_k}=0$ is the best-response for $v_k\in \VV_i$ and hence, it won't deviate. Hence, $\bar{x}$ is a PSNE in the heterogeneous BNPG game on \GG.
\end{proof}
} 

\subsection{XP Algorithm for the parameter treewidth}
Our next result is an \xp algorithm for the \psne problem when parameterized by treewidth.  Towards that, we introduce the notion of ``feasible function'' in \Cref{def:fb} and prove a related algorithmic result in \Cref{feasible:D}.

\begin{definition}\label{def:fb}
Let \GG=(\VV,\EE) be a graph with maximum degree $\Delta$. Let $f:V\rightarrow [\Delta]\cup\{0\}$ be a function where $V\subseteq \VV$. We call a function $f$ feasible if there exists a strategy profile $S$ of all the players in \GG  such that for each $u\in V$, number of neighbours of $u$ playing $1$ in the strategy profile $S$ is $f(u)$. 
\end{definition}

\begin{lemma}[$\star$]\label{feasible:D}
Let \GG=(\VV,\EE) be a graph with maximum degree $\Delta$. Let $V\subseteq \VV$. Then the set of all feasible functions $f:V\rightarrow [\Delta]\cup\{0\}$ can be computed in time $\OO^*(\Delta^{|V|})$.
\end{lemma}

\longversion{
\begin{proof}
We use dynamic programming to solve this problem. Let $N(V)$ denote set of vertices in \VV which is adjacent to at least one vertex in $V$. Let $N[V]:= V\cup N(V)$. Let $N[V]=\{u_1,\ldots,u_l\}$ where $l$ is at most $O(|V|\cdot \Delta)$. Let $c[(d_u)_{u\in V},i]$ denote whether there exists a strategy profile $S$ such that for each $u\in V$, number of neighbours of $u$ in the set $\{u_1,\ldots,u_i\}$ ($\phi$ if $i=0$) playing $1$ in the strategy profile $S$ is $d_u$. Let $f$ be a function such that $f(u)=d_u$ for all $u\in V$. Clearly, $c[(d_u)_{u\in V},l]$ indicates whether the function $f$ is feasible or not. Let $g:\VV\times \VV\rightarrow \{0,1\}$ be a function such that $g(\{u,v\})=1$ if and only if $\{u,v\}\in \EE$. We now present the recursive equation to compute $c[(d_u)_{u\in V},i]$:
        \[      
         c[(d_u)_{u\in V},i]=\begin{cases}
	0 \text{ if }\exists u,d_u<0\\
	c[(d_{u}- g(\{u,u_i\}))_{u\in V},i-1] \vee c[(d_u)_{u\in V},i-1] \text{ if } i\geq1 \\
	1  \text{ if $i=0$ and }\forall u\in V,d_u=0 \\
	0 \text{ otherwise }
	\end{cases}
	\]

Now we argue for the correctness of the above recursive equation. Base cases are trivial as it follows from the definitions. We now look at the case where $i\geq 1$ and $\forall u\in V, d_u\geq 0$. If $c[(d_{u}- g(\{u,u_i\}))_{u\in V},i-1]=1$, then consider the strategy profile $S$ which makes it $1$. Now consider a strategy profile where the response of $u_i$ is $1$ and rest of the players play as per $S$. So for any $u\in V$, the number of neighbours in $\{u_1,\ldots,u_i\}$ playing $1$ increases by $g(\{u,u_i\})$ when compared to the number of neighbours in $\{u_1,\ldots,u_{i-1}\}$ playing $1$. This would imply that $c[(d_u)_{u\in V},i]=1$. Similarly, if $ c[(d_u)_{u\in V},i-1]=1$ ,  then considering the response of $u_i$ as $0$ will not change the number of neighbours playing $1$ and therefore $c[(d_u)_{u\in V},i]=1$. In the other direction, if $c[(d_u)_{u\in V},i]=1$, then consider the strategy profile $S$ which makes it $1$. If the response of $u_i$ is $1$ (resp. $0$) in $S$, then the number of neighbours of $u$ in $\{u_1,\ldots,u_{i-1}\}$ playing $1$ decreases by $g(\{u,u_i\}$ (resp. $0$) when compared with the number of neighbours in $\{u_1,\ldots,u_{i}\}$ playing $1$. Hence, $c[(d_{u}- g(\{u,u_i\}))_{u\in V},i-1] \vee c[(d_u)_{u\in V},i-1]$ is equal to $1$.

Now we look at the time complexity. Total number of cells in the dynamic programming table which we created is $\OO^*(\Delta^{|V|})$ as the value of each entry in $(d_u)_{u\in V}$ is at most $\Delta$. Time spent in each cell is $n^{O(1)}$. Hence, the set of all feasible functions $f:V\rightarrow [\Delta]\cup\{0\}$ can be computed in time $\OO^*(\Delta^{|V|})$.
\end{proof}
}
We now present a $\OO^*(\Delta^{O(k)})$ time \xp algorithm for \psne where $k$ is the treewidth of the input graph. Note that the running time of $\OO^*(\Delta^{O(k)})$ implies that \psne is fixed-parameter tractable for the combined parameter ``treewidth+maximum degree''.
\begin{theorem}\label{thm:XP}
Let \GG be an n-vertex graph given together with its tree decomposition of treewidth at most $k$. Then there is an algorithm running in time $\OO^*(\Delta^{O(k)})$ for \psne in BNPG game on $\GG$ where $\Delta$ is the maximum degree of graph $\GG$.
\end{theorem}
\longversion{\begin{proof}}
\shortversion{
\begin{proof}[Proof Sketch]}
Let $(\GG=(\VV,\EE),(g_v)_{v\in\VV},(c_v)_{v\in\VV})$ be any instance of \psne for BNPG games. Let $(\beta_v(.))_{v\in\VV}$ be the set of the best response functions. Let $\TT = (T, \{X_t \}_{t\in V (T )} )$ be a nice tree decomposition of the input
$n$-vertex graph $\GG$ that has width at most $k$. Let \TT be rooted at some node $r$. For a node $t$ of \TT , let $V_t$ be the union of all the bags present in the subtree of \TT rooted at $t$, including $X_t$. We solve the \psne problem using dynamic programing. Let $N_1(X_t)$ denote set of vertices in $\VV\setminus V_t$ which is adjacent to at least one vertex in $X_t$. Let $N_2(X_t)$ denote set of vertices in $V_t\setminus X_t$ which is adjacent to at least one vertex in $X_t$.Let $c[t,(x_v)_{v\in{X_t}},(d_v^1)_{v\in{X_t}},(d_v^2)_{v\in{X_t}}] =1$ (resp. 0) denote that there exists (resp. doesn't exist) a strategy profile $S$ of all the players in \GG such that for each $u\in X_t$, $u$ plays $x_u$, number of neighbours of $u$ in $N_1(X_t)$ (resp. $N_2(X_t)$) playing $1$ is $d_u^1$ (resp. $d_u^2$) and none of the vertices in $V_t$ deviate in the strategy profile $S$. Before we proceed, we would like to introduce some notations. Let $V$ be a set of vertices and $S_1=(x_v)_{v\in V}$, $S_2=(x_v)_{v\in V\setminus\{w\}}$ be two tuples. Then $S_1\setminus\{x_w\}:=S_2$ and $S_2\cup\{x_w\}:=S_1$. Also, we denote an empty tuple by $\phi$. Clearly $c[r,\phi,\phi,\phi]$ indicates whether there is a PSNE in \GG or not. We now present the recursive equation to compute $c[t,(x_v)_{v\in{X_t}},(d_v^1)_{v\in{X_t}},(d_v^2)_{v\in{X_t}}]$ for various types of node in $\TT$.

\textbf{Leaf Node:} For a leaf node $t$ we have that $X_t = \phi$. Hence, $c[t,\phi,\phi,\phi]=1$.

\textbf{Join Node:} For a join node $t$, let $t_1,t_2$ be its two children. Note that $X_t=X_{t_1}=X_{t_2}$. 

Now we proceed to compute $c[t,(x_v)_{v\in{X_t}},(d_v^1)_{v\in{X_t}},(d_v^2)_{v\in{X_t}}]$. %If there is a vertex $v\in V_t$ such that $x_v\notin \beta(n_v'+d_v^1+d_v^2)$ where $n_v'$ is the number of neighbours of $v$ in $X_t$ playing  $1$ in $(x_v)_{v\in{X_t}}$, then $c[t,(x_v)_{v\in{X_t}},(d_v^1)_{v\in{X_t}},(d_v^2)_{v\in{X_t}}]=0$. 
Let $\mathcal{F}$ be a set of tuples $(d_v')_{v\in X_t}$ such that there is a strategy profile $S$ such that for each $v\in X_t$, its response is $x_v$, the number of neighbours in $N_1(x)$, $V_{t_1}\setminus X_{t_1}$ and $V_{t_2}\setminus X_{t_2}$ playing $1$ is $d_v^1,d_v',d_v^2-d_v'$ respectively. Using \Cref{feasible:D} we can find the set $\mathcal{F}$ in time $\OO^*(\Delta^k)$. Then $c[t,(x_v)_{v\in{X_t}},(d_v^1)_{v\in{X_t}},(d_v^2)_{v\in{X_t}}]$ is equal to the following formula:
\begin{align*}
&0\vee\bigvee_{(d_v')_{v\in X_t}\in \mathcal{F}}\big(c[t_1,(x_v)_{v\in{X_t}},(d_v^1+d_v^2-d_v')_{v\in{X_t}},(d_v')_{v\in{X_t}}]\\
&\wedge c[t_2,(x_v)_{v\in{X_t}},(d_v^1+d_v')_{v\in{X_t}},(d_v^2-d_v')_{v\in{X_t}}]\big)
\end{align*}
\longversion{
Now we argue for the correctness of the above recursive equation. If $\mathcal{F}$ is empty then the above equation trivially holds true. Hence, we assume that $\mathcal{F}$ is non-empty. In one direction let us assume that there exists a tuple $(d_v')_{v\in X_t}\in \mathcal{F}$ such that $(c[t_1,(x_v)_{v\in{X_t}},(d_v^1+d_v^2-d_v')_{v\in{X_t}},(d_v')_{v\in{X_t}}]\wedge c[t_2,(x_v)_{v\in{X_t}},(d_v^1+d_v')_{v\in{X_t}},(d_v^2$$-d_v')_{v\in{X_t}}])$=1. Let $S_1$ and $S_2$ be the strategy profiles which leads to $c[t_1,(x_v)_{v\in{X_t}},(d_v^1+d_v^2-d_v')_{v\in{X_t}},(d_v')_{v\in{X_t}}]$ and $c[t_2,(x_v)_{v\in{X_t}},(d_v^1+d_v')_{v\in{X_t}},(d_v^2$$-d_v')_{v\in{X_t}}]$ respectively being $1$. Let $S_3$ be a strategy profile that leads to $(d_v')_{v\in X_t}$ being included in $\mathcal{F}$. Now consider a strategy profile $S$ such that responses of players in $V_{t_1}\setminus X_{t}$ is their responses in $S_1$, responses of the players in $V_{t_2}\setminus X_{t}$ is their responses in $S_2$ and responses of rest of the players is their responses in $S_3$. Now observe that there are no edges between the set of vertices $V_{t_1}\setminus X_t$ and $V_{t_2}\setminus X_t$ and there are no edges between the set of vertices $\VV\setminus V_t$ and $V_{t}\setminus X_t$. Therefore, the number of neighbours of vertices in $V_{t_1}\setminus X_{t_1}$ and $V_{t_2}\setminus X_{t_2}$ playing $1$ doesn't change when compared to $S_1$ and $S_2$ respectively. Also, the number of neighbours of vertices in $X_t$ playing $1$ doesn't change when compared to $S_1$, $S_2$ and $S_3$. Hence, $c[t,(x_v)_{v\in{X_t}},(d_v^1)_{v\in{X_t}},(d_v^2)_{v\in{X_t}}]=1$. In other direction, let $c[t,(x_v)_{v\in{X_t}},(d_v^1)_{v\in{X_t}},(d_v^2)_{v\in{X_t}}]=1$. Let $S'$ be the strategy profile which leads to $c[t,(x_v)_{v\in{X_t}},(d_v^1)_{v\in{X_t}},(d_v^2)_{v\in{X_t}}]$ being $1$. For each $v\in X_t$, let the number of neighbours in $V_{t_1}\setminus X_t$ playing $1$ in $S'$ be $d_v'$. Then clearly $S'$ leads to both $c[t_1,(x_v)_{v\in{X_t}},(d_v^1+d_v^2-d_v')_{v\in{X_t}},(d_v')_{v\in{X_t}}]$ and $c[t_2,(x_v)_{v\in{X_t}},(d_v^1+d_v')_{v\in{X_t}},(d_v^2$$-d_v')_{v\in{X_t}}]$ being $1$
}

\textbf{Introduce Node:} Let $t$ be an introduce node with a child $t'$ such that $X_t = X_{t'} \cup \{u\}$ for some $u \notin X_{t'}$. Let $S'=(x_v)_{v\in{X_t}}$ be a strategy profile of vertices in $X_t$. Let $n_v'$ denote the number of neighbours of $v$ playing $1$ in $S'$.  Let $g:\VV\times \VV\rightarrow \{0,1\}$ be a function such that $g(\{u,v\})=1$ if and only if $\{u,v\}\in \EE$.  We now proceed to compute $c[t,S',(d_v^1)_{v\in{X_t}},(d_v^2)_{v\in{X_t}}]$. If there is no strategy profile $S$ where $\forall v\in X_t$,the number of neighbours of $v$ in $N_1(X_t)$ (resp. $N_2(X_t)$) playing $1$ is $d_v^1$ (resp. $d_v^2$) , then clearly $c[t,S',(d_v^1)_{v\in{X_t}},(d_v^2)_{v\in{X_t}}]=0$. Due to \Cref{feasible:D}, we can check the previous statement in $\OO^*(\Delta^{k})$ by considering a bipartite subgraph of \GG between $X_t$ and $N_1(X_t)$ (or $N_2(X_t)$).  Otherwise, we have the following:

\[	
	c[t,S',(d_v^1)_{v\in{X_t}},(d_v^2)_{v\in{X_t}}]=
	\begin{cases}
		0 \text{ if }\exists v\in X_t, x_v\notin\beta_v(n_v'+d_v^1+d_v^2)\\
		c[t',S'\setminus \{x_u\},(d_v^1+g(\{v,u\}))_{v\in{X_{t'}}},(d_v^2)_{v\in{X_{t'}}}] \text{ if }x_u=1\\
		c[t',S'\setminus \{x_u\},(d_v^1)_{v\in{X_{t'}}},(d_v^2)_{v\in{X_{t'}}}] \text{ otherwise }
	\end{cases}
	\]
\longversion{	
Now we argue for the correctness of the above recursive equation. The base case is trivial. Now if $x_u=1$ and $c[t',S'\setminus \{x_u\},(d_v^1+g(\{v,u\}))_{v\in{X_{t'}}},(d_v^2)_{v\in{X_{t'}}}]=1$, then let $S_1$ be the strategy profile that makes the cell value to be $1$.
Let $S_2$ be a strategy such that $\forall v\in X_t$, the number of neighbours of $v$ in $N_1(X_t)$ (resp. $N_2(X_t)$) playing $1$ is $d_v^1$ (resp. $d_v^2$). Let $S_3$ be a strategy profile where the responses of players in $X_t$ is $(x_v)_{v\in X_t}$, the responses of players in $V_t\setminus X_t$ is their responses in $S_1$ and the responses of the rest of the players is their responses in $S_2$. Now observe that there is no edge between the set of vertices $\VV\setminus V_{t'}$ and $V_{t'}\setminus X_{t'}$ (due to this the only valid value of $d_u^2$ is 0 which is ensured by us). Therefore none of the vertices deviate in $V_{t'}\setminus X_{t'}$ deviate. Also, the number of neighbours of the vertices in $X_{t'}$ playing $1$ in $S_3$ is the same as that of $S_1$. Hence, the vertices in $X_{t'}$ don't deviate. Vertex $u$ doesn't deviate as we are not in the base case and therefore $x_u\in\beta_v(n_u'+d_u^1+d_u^2)$. Hence, $c[t,S',(d_v^1)_{v\in{X_t}},(d_v^2)_{v\in{X_t}}]=1$. By similar arguments we can show that if $x_u=0$ and $c[t',S'\setminus \{x_u\},(d_v^1))_{v\in{X_{t'}}},(d_v^2)_{v\in{X_{t'}}}]=1$ then $c[t,S',(d_v^1)_{v\in{X_t}},(d_v^2)_{v\in{X_t}}]=1$. In the other direction, if $c[t,S',(d_v^1)_{v\in{X_t}},(d_v^2)_{v\in{X_t}}]=1$, then let $S_4$ be the strategy profile which makes the value of this cell $1$. If $x_u=1$ in $S_4$, then for each $v\in X_{t'}$, the number of neighbours of $X_{t'}$ in $N_1(X_{t'})$ (resp. $N_2(X_{t'}))$ playing $1$ in $S_4$ is $d_v^1+g(\{v,u\})$ (resp. $d_v^2$). Since none of the vertices in $V_t$ deviate, so clearly $S_4$ leads $c[t',S'\setminus \{x_u\},(d_v^1+g(\{v,u\}))_{v\in{X_{t'}}},(d_v^2)_{v\in{X_{t'}}}]$ to being $1$ . Similar argument holds when $x_u=0$ in $S_4$ , and it would lead $c[t',S'\setminus \{x_u\},(d_v^1)_{v\in{X_{t'}}},(d_v^2)_{v\in{X_{t'}}}]$ to being $1$.
}

\textbf{Forget Node:} Let $t$ be a forget node with a child $t'$ such that $X_t$ = $X_t' \setminus \{w\}$ for some $w \in X_{t'}$. Let $S_0=(x_v)_{v\in{X_t}}\cup \{x_w=0\},$ $S_1=(x_v)_{v\in{X_t}}\cup \{x_w=1\}$ be two strategy profiles of vertices in $X_t'$. Let $g:\VV\times \VV\rightarrow \{0,1\}$ be a function such that $g(\{u,v\})=1$ if and only if $\{u,v\}\in \EE$. We now have the following:
\begin{align*}
c[t,(x_v)_{v\in X_t},(d_v^1)_{v\in{X_t}},(d_v^2)_{v\in{X_t}}]=&\bigvee_{d_w^1,d_w^2:0\leq d_w^1,d_w^2\leq \Delta}\big(c[t',S_0,(d_v^1)_{v\in{X_{t'}}},(d_v^2)_{v\in{X_{t'}}}]\\
&\vee c[t',S_1,(d_v^1)_{v\in{X_{t'}}},(d_v^2-g(\{v,w\}))_{v\in{X_{t'}}}\big)   	
\end{align*}
\longversion{
Now we argue for the correctness of the above recursive equation. In one direction, let us assume that $\exists d_w^1,d_w^2\in [\Delta]\cup \{0\}$ such that $(c[t',S_0,(d_v^1)_{v\in{X_{t'}}},(d_v^2)_{v\in{X_{t'}}}]\vee c[t',S_1,(d_v^1)_{v\in{X_{t'}}},(d_v^2-g(\{v,w\}))_{v\in{X_{t'}}}])=1$. Let $S$ be the strategy profile that made the formula in the previous statement to be true. Now observe that for each $v\in X_t$, the number of neighbours in $N_1(X_t)$ and $N_2(X_t)$ is $d^1_v$ and $d^2_v$ respectively. Since none of the vertices in the set $V_t=V_{t'}$ deviate, therefore $c[t,(x_v)_{v\in{N[X_t]}},(d_v^1)_{v\in{X_t}},(d_v^2)_{v\in{X_t}}]=1$. In other direction, let us assume that $c[t,(x_v)_{v\in{N[X_t]}},(d_v^1)_{v\in{X_t}},(d_v^2)_{v\in{X_t}}]=1$. Let $S'$ be the strategy profile that made the formula in the previous statement to be true. Let $d_w^1$ and $d_w^2$ be the number of neighbours of $w$ in $N_1(X_{t'})$ and $N_2(X_{t'})$ respectively playing $1$ in the profile $S'$. For each $v\in X_{t}$, if $x_w=1$ (resp. $x_w=0$) then the number of neighbours in $N_2(X_{t'})$ playing $1$ in $S'$ is $d_v^2-g(\{v,w\}$ (resp. $d_v^2$).  Similarly, for each $v\in X_{t}$, the number of neighbours in $N_1(X_{t'})$ playing $1$ in $S'$ is $d_v^1$. Since none of the vertices in $V_t=V_{t'}$ deviate in $S'$, therefore   $(c[t',S_0,(d_v^1)_{v\in{X_{t'}}},(d_v^2)_{v\in{X_{t'}}}]\vee c[t',S_1,(d_v^1)_{v\in{X_{t'}}},(d_v^2-g(\{v,w\}))_{v\in{X_{t'}}}])=1$. \\
}
\shortversion{We refer the reader to the appendix for the proof of correctness of the above recursive equations.} Now we consider the time complexity of our algorithm. Total number of cells in the dynamic programming table which we created is $\OO^*(\Delta^{O(k)})$. For each cell, we spend at most $\OO^*(\Delta^{O(k)})$ time if we are computing the table in a bottom up fashion. Hence, the running time is $\OO^*(\Delta^{O(k)})$. 
\end{proof}

\subsection{Tractable Results}

%This finishes our hardness results.
 To conclude our fine-grained analysis of the \psne problem, we bridge the gap between the tractablility and intractibility by showing some tractable results. Our first result is an \fpt algorithm for \psne for strict games when parameterized by the vertex cover number. 
%The high-level idea of our algorithm is the following. For every possible set of responses of players in a minimum vertex cover \WW of the instance graph \GG, there is only one best response for the players not in \WW. This is due to the fact that we are looking at a strict BNPG game and the set of nodes not in \WW forms an independent set. We iterate over all possible responses of players in \WW and check the existence of PSNE.

\begin{theorem}[$\star$]\label{vertex-cover1}
	There is a $\OO^*(2^\vc)$ time algorithm for \psne for strict BNPG games where \vc is the vertex cover number.
\end{theorem}

\longversion{
\begin{proof}
	Let $(\GG=(\VV,\EE),(g_v)_{v\in\VV},(c_v)_{v\in\VV})$ be any instance of \psne for BNPG games. We compute a minimum vertex cover $\SS\subset\VV$ in time $\OO^*(2^{\vc})$~\cite{CyganFKLMPPS15}. The idea is to brute force on the strategy profile of players in \SS and assign actions of other players based on their best-response functions. For every strategy profile $x_\SS=(x_v)_{v\in\SS}$, we do the following.
	
	\begin{enumerate}
		\item For $w\in\VV\setminus\SS$, let $n_w$ be the number of neighbors of $w$ (they can only be in \SS) who play $1$. We define $x_w = 1$ if  $\Delta g_{w}(n_w)> c_{w}$ and $0$ if $\Delta g_{w}(n_w)< c_{w}$. This is well-defined since $\Delta g_{w}(n_w)\ne c_{w}$ as the game is strict.
		\item If $(x_v)_{v\in \VV}$ forms a PSNE, then output \yes. Otherwise, we discard the current $x_\SS$.
	\end{enumerate}
	If the above the procedure does not output \yes for any $x_\SS$, then we output \no. The correctness of the algorithm is immediate. Since the computation for every guess of $x_\SS$ can be done in polynomial time and the number of such guesses is $2^\SS=2^\vc$, it follows that the running time of our algorithm is $\OO^*(2^\vc)$.
\end{proof}
}

%We also observe that the \psne problem for BNPG games, parameterized by both maximum degree and diameter, is fixed-parameter-tractable as the number of vertices is upper bounded by $\Delta^d$. This complements our hardness result of \Cref{thm:dia}.

%\begin{observation}
%The \psne problem for BNPG games is $FPT$ with respect to the parameter $(\Delta,d)$.
%\end{observation}

%\longversion{
%\begin{proof}
	%Let \GG be any graph with diameter $d$ and maximum degree $\Delta$. A simple breadth-first search based argument proves that the number of vertices $n$ of a graph is at most $\Delta^d$. Hence, we can run a brute force search in time $O(2^{O(\Delta^d)}\cdot O(\Delta^{d\cdot O(1)}))$ by checking every strategy profile and conclude whether there is any PSNE or not. 
%\end{proof}
%}

Our next result shows that we can always find a PSNE for additive BNPG games in $\OO(n)$ time. This complements the intractable result for subadditive BNPG games.

\begin{observation}[$\star$]\label{obs:add}
	There exists an $\OO(n)$ time algorithm to find a PSNE in an additive BNPG game.
\end{observation}

\longversion{
\begin{proof}
$\forall x\geq 0,\forall i\in[n]$, $g_{v_i}(x+1)-g_{v_i}(x)=g_{v_i}(1)$. This implies for a player $v_i$, the best response doesn't depend on the responses of its neighbours and solely depends on $g_{v_i}(1)$. Hence, if $g_{v_i}(1)\geq c_{v_i}$ then we assign the response of player $v_i$ as 1 and 0 otherwise. This will make sure that no player $v_i$ deviates. So calculating the PSNE takes $\OO(n)$
\end{proof}
}

We next consider circuit rank and distance from complete graph as parameter. These parameters can be thought of distance from tractable instances (namely tree and complete graph). They are defined as follows.

\begin{definition}\label{def:d1}
Let the number of edges and number of vertices in a graph $\mathcal{G}$ be $m$ and $n$ respectively. Then $d_{1}$ (circuit rank) is defined to be $m-n+c$ ($c$ is the number of connected components in the graph) and $d_2$ (distance from complete graph) is defined to be $\frac{n(n-1)}{2}-m$. Note that circuit rank is not the same as feedback arc set.
\end{definition}

Yu et al. presented an algorithm for \psne on trees in \cite{yu2020computing}. It turns out that their algorithm can be appropriately modified to get the following observation.

\begin{observation}\cite{yu2020computing}\label{obs:tree}
	Given a BNPG game on a tree $\TT=(\VV,\EE)$, a subset of vertices $\UU\subseteq\VV$ and a strategy profile $(x_u)_{u\in\UU}\in\{0,1\}^\UU$, there is a polynomial time algorithm for deciding if there exists a PSNE $(y_v)_{v\in\VV}\in\{0,1\}^\VV$ for the BNPG game such that $x_u=y_u$ for every $u\in\UU$.
\end{observation}

Now by using the observation \ref{obs:tree} as a subroutine, we exhibit an \fpt algorithm for the parameter circuit rank.

\begin{theorem}\label{thm:d1}
There is an algorithm running in time $\OO^*(4^{d_1})$ for \psne in BNPG games where $d_1$ is the circuit rank of the input graph. 
\end{theorem}
\begin{proof}
	Let $(\GG=(\VV,\EE),(g_v)_{v\in\VV},(c_v)_{v\in\VV})$ be any instance of \psne for BNPG games. Let the graph \GG have $c$ connected components namely, $\GG_1=(\VV_1,\EE_1),\ldots,\GG_c=(\VV_c,\EE_c)$. For every $i\in[c]$, we decide if there exists a PSNE in $\GG_i$; clearly there is a PSNE in \GG if and only if there is a PSNE in $\GG_i$ for every $i\in[n]$. Hence, in the rest of the proof, we focus on the algorithm to decide the existence of a PSNE in $\GG_i$. We compute a minimum spanning tree $\TT_i$ in the connected component $\GG_i$. Let $\EE_i^\pr\subset\EE_i$ be the set of edges which are not part of $\TT_i$; let $|\EE_i^\pr|=d_1^i$ and $\VV_i^\pr=\{v_1^i,v_2^i,\ldots,v_l^i\}\subseteq\VV_i$ be the set of vertices which are endpoints of at least one edge in $\EE_i^\pr$. Of course, we have $|\VV_i^\pr|=\el\leq 2d_1^i$. For every tuple $t=(x_v^\pr)_{v\in\VV_i^\pr}\in\{0,1\}^l$, we do the following.
	
	\begin{enumerate}
		\item For each $v\in\VV_i^\pr$, let $n_v^t$ be the number of neighbours of $v$ in $\GG_i[\EE^\pr_i]$ ( subgraph of $\GG_i$ containing the set of nodes $\VV_i$ and the set of edges $\EE_i^\pr$ ) who play $1$ in $t$. We now define $g^t_v$ for every player $v\in\VV$ as follows.
		
		\[
		g^t_v(k) =\begin{cases}
			g_v(k+n_v^t) & \text{if } v\in\VV_i^\pr\\
			g_v(k) & \text{otherwise}
		\end{cases}
		\]
		
		\item We now decide if there exists a PSNE $(y_v)_{v\in\VV_i}\in\{0,1\}^{\VV_i}$ in the BNPG game $(\TT_i, (g_v^t)_{v\in\VV_i}, (c_v)_{v\in\VV_i})$ such that $y_v=x^\pr_v$ for every $v\in\VV_i^\pr$; this can be done in polynomial time due to \Cref{obs:tree}. If such a PSNE exists, then we output \yes.
	\end{enumerate}

	If we fail to find a PSNE for every choice of tuple $t$, then we output \no. The running time of the above algorithm (for $\GG_i$) is $\OO^*\left(2^{|\VV_i^\pr|}\right)$. Hence the overall running time of our algorithm is $\OO^*\left(\sum_{i=1}^c2^{|\VV_i^\pr|}\right) \leq \OO^*\left(2^{2d_1}\right)=\OO^*\left(4^{d_1}\right)$. We now argue correctness of our algorithm. We observe that it is enough to argue correctness for one component.
	
	In one direction, let $x^*=(x_v^*)_{v\in\VV_i}$ be a PSNE in the BNPG game $(\GG_i, (g_v)_{v\in\VV_i},(c_v)_{v\in\VV_i})$. We now claim that $(x_v^*)_{v\in\VV_i}$ is also a PSNE in the BNPG game on $(\TT_i, (g^t_v)_{v\in\VV_i}, (c_v)_{v\in\VV_i})$ where $t=(x_v^*)_{v\in\VV_i^\pr}$. Let $n_v^{\GG_i}$ and $n_v^{\TT_i}$ be the number of neighbors of $v\in\VV_i$ in $\GG_i$ and $\TT_i$ respectively who play $1$ in $x^*$. With $n^t_v$ defined as above, we have $n^{\GG_i}_v=n^{\TT_i}_v+n_v^t$ for $v\in \VV_i^\pr$ and $n^{\GG_i}_v=n^{\TT_i}_v$ for $v\in\VV_i\setminus\VV_i^\pr$. Hence, we have $\Delta g^t_v(n^{\TT_i}_v)=\Delta g_v(n^{\TT_i}_v+n_v^t)=\Delta g_v(n^{\GG_i}_v)$ for $v\in \VV_i^\pr$ and $\Delta g^t_v(n^{\TT_i}_v)=\Delta g_v(n^{\GG_i}_v)$ for $v\in\VV_i\setminus\VV_i^\pr$. If $x_v^*=1$ where $v\in \VV_i$, then $\Delta g_v(n^{\GG_i}_v)\geq c_v$ and thus we have $\Delta g^t_v(n^{\TT_i}_v)\geq c_v$. Hence, $v$ does not deviate in $\TT_i$. Similarly, if $x_v^*=0$ where $v\in \VV_i$, then $\Delta g_v(n^{\GG_i}_v)\leq c_v$ and thus we have $\Delta g^t_v(n^{\TT_i}_v)\leq c_v$. Hence, $v$ does not deviate in $\TT_i$. Hence $(x_v^*)_{v\in\VV}$ is also a PSNE in BNPG game $(\TT_i, (g^t_v)_{v\in\VV_i}, (c_v)_{v\in\VV_i})$ where $t=(x_v^*)_{v\in\VV_i^\pr}$ (which means our Algorithm returns YES). 
	
	In the other direction, let $(x_v^*)_{v\in\VV_i}$ be the PSNE in BNPG game on $(\TT_i, (g^t_v)_{v\in\VV_i}, (c_v)_{v\in\VV_i})$ where $t=(x_v^*)_{v\in\VV_i^\pr}$ (which means our Algorithm returns YES). We claim that $(x_v^*)_{v\in\VV_i}$ is also a PSNE in BNPG game $(\GG_i, (g_v)_{v\in\VV_i},(c_v)_{v\in\VV_i})$. If $x_v^*=1$ for $v\in \VV_i$, then $\Delta g^t_v(n^{\TT_i}_v)\geq c_v$. This implies that $\Delta g_v(n^{\GG_i}_v)\geq c_v$ and thus $v$ does not deviate in $\GG_i$. Similarly, if $x_v^*=0$ for $v\in \VV_i$, then $\Delta g^t_v(n^{\TT_i}_v)\leq c_v$. This implies that $\Delta g_v(n^{\GG_i}_v)\leq c_v$ and thus $v$ does not deviate in $\GG_i$. Hence $(x_v^*)_{v\in\VV}$ is also a PSNE in BNPG game on $(\GG_i, (g_v)_{v\in\VV_i},(c_v)_{v\in\VV_i})$.
\end{proof}

\longversion{
\begin{proof}
	Let $(\GG=(\VV,\EE),(g_v)_{v\in\VV},(c_v)_{v\in\VV})$ be any instance of \psne for BNPG games. Let the graph \GG have $c$ connected components namely, $\GG_1=(\VV_1,\EE_1),\ldots,\GG_c=(\VV_c,\EE_c)$. For every $i\in[c]$, we compute if there exists a PSNE in $\GG_i$; clearly there is a PSNE in \GG if and only if there is a PSNE in $\GG_i$ for every $i\in[n]$. Hence, in the rest of the proof, we focus on the algorithm to compute a PSNE in $\GG_i$. We compute a minimum spanning tree $\TT_i$ in the connected component $\GG_i$. Let $\EE_i^\pr\subset\EE_i$ be the set of edges which are not part of $\TT_i$; let $|\EE_i^\pr|=d_1^i$ and $\VV_i^\pr=\{v_1^i,v_2^i,\ldots,v_l^i\}\subseteq\VV_i$ be the set of vertices which are endpoints of at least one edge in $\EE_i^\pr$. Of course, we have $|\VV_i^\pr|=\el\leq 2d_1^i$. For every tuple $t=(x_v^\pr)_{v\in\VV_i^\pr}\in\{0,1\}^l$, we do the following.
	
	\begin{enumerate}
		\item For each $v\in\VV_i^\pr$, let $n_v^t$ be the number of neighbours of $v$ in $\GG_i[\EE^\pr_i]$ ( subgraph of $\GG_i$ containing the set of nodes $\VV_i$ and the set of edges $\EE_i^\pr$ ) who play $1$ in $t$. We now define $g^t_v$ for every player $v\in\VV$ as follows.
		
		\[
		g^t_v(k) =\begin{cases}
			g_v(k+n_v^t) & \text{if } v\in\VV_i^\pr\\
			g_v(k) & \text{otherwise}
		\end{cases}
		\]
		
		\item We now compute if there exists a PSNE $(y_v)_{v\in\VV_i}\in\{0,1\}^{\VV_i}$ in the BNPG game $(\TT_i, (g_v^t)_{v\in\VV_i}, (c_v)_{v\in\VV_i})$ such that $y_v=x^\pr_v$ for every $v\in\VV_i^\pr$; this can be done in polynomial time due to \Cref{obs:tree}. If such a PSNE exists, then we output \yes.
	\end{enumerate}

	If we fail to find a PSNE for every choice of tuple $t$, then we output \no. The running time of the above algorithm (for $\GG_i$) is $\OO^*\left(2^{|\VV_i^\pr|}\right)$. Hence the overall running time of our algorithm is $\OO^*\left(\sum_{i=1}^c2^{|\VV_i^\pr|}\right) \leq \OO^*\left(2^{2d_1}\right)=\OO^*\left(4^{d_1}\right)$. We now argue correctness of our algorithm. We observe that it is enough to argue correctness for one component.
	
	In one direction, let $x^*=(x_v^*)_{v\in\VV_i}$ be a PSNE in the BNPG game $(\GG_i, (g_v)_{v\in\VV_i},(c_v)_{v\in\VV_i})$. We now claim that $(x_v^*)_{v\in\VV_i}$ is also a PSNE in the BNPG game on $(\TT_i, (g^t_v)_{v\in\VV_i}, (c_v)_{v\in\VV_i})$ where $t=(x_v^*)_{v\in\VV_i^\pr}$. Let $n_v^{\GG_i}$ and $n_v^{\TT_i}$ be the number of neighbors of $v\in\VV_i$ in $\GG_i$ and $\TT_i$ respectively who play $1$ in $x^*$. With $n^t_v$ defined as above, we have $n^{\GG_i}_v=n^{\TT_i}_v+n_v^t$ for $v\in \VV_i^\pr$ and $n^{\GG_i}_v=n^{\TT_i}_v$ for $v\in\VV_i\setminus\VV_i^\pr$. Hence, we have $\Delta g^t_v(n^{\TT_i}_v)=\Delta g_v(n^{\TT_i}_v+n_v^t)=\Delta g_v(n^{\GG_i}_v)$ for $v\in \VV_i^\pr$ and $\Delta g^t_v(n^{\TT_i}_v)=\Delta g_v(n^{\GG_i}_v)$ for $v\in\VV_i\setminus\VV_i^\pr$. If $x_v^*=1$ where $v\in \VV_i$, then $\Delta g_v(n^{\GG_i}_v)\geq c_v$ and thus we have $\Delta g^t_v(n^{\TT_i}_v)\geq c_v$. Hence, $v$ does not deviate in $\TT_i$. Similarly, if $x_v^*=0$ where $v\in \VV_i$, then $\Delta g_v(n^{\GG_i}_v)\leq c_v$ and thus we have $\Delta g^t_v(n^{\TT_i}_v)\leq c_v$. Hence, $v$ does not deviate in $\TT_i$. Hence $(x_v^*)_{v\in\VV}$ is also a PSNE in BNPG game $(\TT_i, (g^t_v)_{v\in\VV_i}, (c_v)_{v\in\VV_i})$ where $t=(x_v^*)_{v\in\VV_i^\pr}$ (which means our Algorithm returns YES). 
	
	In the other direction, let $(x_v^*)_{v\in\VV_i}$ be the PSNE in BNPG game on $(\TT_i, (g^t_v)_{v\in\VV_i}, (c_v)_{v\in\VV_i})$ where $t=(x_v^*)_{v\in\VV_i^\pr}$ (which means our Algorithm returns YES). We claim that $(x_v^*)_{v\in\VV_i}$ is also a PSNE in BNPG game $(\GG_i, (g_v)_{v\in\VV_i},(c_v)_{v\in\VV_i})$. If $x_v^*=1$ for $v\in \VV_i$, then $\Delta g^t_v(n^{\TT_i}_v)\geq c_v$. This implies that $\Delta g_v(n^{\GG_i}_v)\geq c_v$ and thus $v$ does not deviate in $\TT_i$. Similarly, if $x_v^*=0$ for $v\in \VV_i$, then $\Delta g^t_v(n^{\TT_i}_v)\leq c_v$. This implies that $\Delta g_v(n^{\GG_i}_v)\leq c_v$ and thus $v$ does not deviate in $\TT_i$. Hence $(x_v^*)_{v\in\VV}$ is also a PSNE in BNPG game on $(\GG_i, (g_v)_{v\in\VV_i},(c_v)_{v\in\VV_i})$.
\end{proof}
}

Yu et al. presented an algorithm for \psne on complete graphs in \cite{yu2020computing}. It turns out that their algorithm can be appropriately modified to get the following observation.
\begin{observation}\cite{yu2020computing}\label{obs:complete}
	Given a BNPG game on a complete graph $\GG=(\VV,\EE)$, and an integer $k$, there is a polynomial time algorithm for deciding if there exists a PSNE where exactly $k$ players play $1$ and returns such a PSNE if it exists.
\end{observation}

Now by using the observation \ref{obs:complete} as a subroutine, we exhibit an \fpt algorithm for the parameter distance from complete graph.

\begin{theorem}\label{thm:d2}
There is an algorithm running in time $\OO^*(4^{d_2})$ for \psne in BNPG games where $d_2$ is the distance from complete graph.
\end{theorem}
\begin{proof}
	Let $(\GG=(\VV,\EE),(g_v)_{v\in\VV},(c_v)_{v\in\VV})$ be any instance of \psne for BNPG games. If $d_2\geq \frac{n}{2}$, then iterating over all possible strategy profiles takes time $\OO^*(2^n)\leq\OO^*(4^{d_2})$. So allow us to assume for the rest of the proof that $d_2<\frac{n}{2}$. Let us define $\VV^\pr=\{u\in\VV: \exists v\in\VV, v\ne u, \{u,v\}\notin\EE\}$; we have $|\VV^\pr|\le 2d_2$. 
	
	For every strategy profile $y=(y_u)_{u\in\VV^\pr}$, we do the following. For each $v\in\VV\setminus\VV^\pr$, let $n_v^\pr$ be the number of neighbors of $v$ in $\VV^\pr$ who play $1$ in $y$. We now define $g^\pr_v(\el)=g_v(\el+n_v^\pr)$ for every $\el\in\NB\cup\{0\}$ and every player $v\in\VV\setminus\VV^\pr$. For every $k\in\{0,\ldots,|\VV\setminus\VV^\pr|\}$, we decide (using the algorithm in \Cref{obs:complete}) if there exists a PSNE $x^k=(x_v^k)_{v\in\VV\setminus\VV^\pr}$ in the BNPG game $(\GG[\VV\setminus\VV^\pr], (g^\pr_v)_{v\in\VV\setminus\VV^\pr}, (c_v)_{v\in\VV\setminus\VV^\pr})$ where exactly $k$ players play $1$. If $x^k$ exists, then we output \yes if $((y_u)_{u\in\VV^\pr},(x^k_v)_{v\in\VV\setminus\VV^\pr})$ forms a PSNE in the BNPG game $(\GG=(\VV,\EE),(g_v)_{v\in\VV},(c_v)_{v\in\VV})$.
	
	If the above procedure fails to find a PSNE, then we output \no. The running time of the above algorithm is $\OO^*\left(2^{|\VV^\pr|}\right)\leq\OO^*\left(4^{d_2}\right)$. We now argue correctness\longversion{ of our algorithm}.
	
	Clearly, if the algorithm outputs \yes, then there exists a PSNE for the input game. On the other hand, if there exists a PSNE $((y_u)_{u\in\VV^\pr},(x_v)_{v\in\VV\setminus\VV^\pr})\in\{0,1\}^\VV$ in the input game, then let us consider the iteration of our algorithm with the guess $(y_u)_{u\in\VV^\pr}$. Let the number of players playing $1$ in $(x_v)_{v\in\VV\setminus\VV^\pr}$ be $k$. If $x_v=1$ where $v\in\VV\setminus\VV^\pr$, then $\Delta g_v(n_v^\pr+k-1)\geq c_v$ and thus we have $\Delta g^\pr_v(k-1)\geq c_v$. Similarly, if $x_v=0$ where $v\in\VV\setminus\VV^\pr$, then $\Delta g_v(n_v^\pr+k)\leq c_v$ and thus we have $\Delta g^\pr_v(k)\leq c_v$. Hence, we observe that $(x_v)_{v\in\VV\setminus\VV^\pr}$ forms a PSNE in the BNPG game $(\GG[\VV\setminus\VV^\pr], (g^\pr_v)_{v\in\VV\setminus\VV^\pr}, (c_v)_{v\in\VV\setminus\VV^\pr})$. Let $(x_v^\pr)_{v\in\VV\setminus\VV^\pr}$ be the PSNE of the BNPG game $(\GG[\VV\setminus\VV^\pr], (g^\pr_v)_{v\in\VV\setminus\VV^\pr}, (c_v)_{v\in\VV\setminus\VV^\pr})$ where exactly $k$ players play $1$ returned by the  algorithm in \Cref{obs:complete}. We observe that every player in $\VV^\pr$ has the same number of neighbors playing $1$ in both the strategy profiles $((y_u)_{u\in\VV^\pr},(x_v)_{v\in\VV\setminus\VV^\pr})$ and $((y_u)_{u\in\VV^\pr},(x_v^\pr)_{v\in\VV\setminus\VV^\pr})$. So no player in $\VV^\pr$ will deviate in the strategy profile $((y_u)_{u\in\VV^\pr},(x_v^\pr)_{v\in\VV\setminus\VV^\pr})$. If $x_v^\pr=1$ where $v\in\VV\setminus\VV^\pr$, then $\Delta g^\pr_v(k-1)\geq c_v$ and thus we have $\Delta g_v(n_v^\pr+k-1)\geq c_v$. Hence, $v$ does not deviate in the strategy profile $((y_u)_{u\in\VV^\pr},(x_v^\pr)_{v\in\VV\setminus\VV^\pr})$. Similarly, if $x_v^\pr=0$ where $v\in\VV\setminus\VV^\pr$, then $\Delta g^\pr_v(k)\leq c_v$ and thus we have $\Delta g_v(n_v^\pr+k)\leq c_v$. Hence, $v$ does not deviate in the strategy profile $((y_u)_{u\in\VV^\pr},(x_v^\pr)_{v\in\VV\setminus\VV^\pr})$. Hence, $((y_u)_{u\in\VV^\pr},(x_v^\pr)_{v\in\VV\setminus\VV^\pr})$ also forms a PSNE in the input BNPG game and thus the algorithm outputs \yes. This concludes the correctness of our algorithm.
\end{proof}

\longversion{
\begin{proof}
	Let $(\GG=(\VV,\EE),(g_v)_{v\in\VV},(c_v)_{v\in\VV})$ be any instance of \psne for BNPG games. If $d_2\geq \frac{n}{2}$, then iterating over all possible strategy profiles takes time $\OO^*(2^n)\leq\OO^*(4^{d_2})$. So allow us to assume for the rest of the proof that $d_2<\frac{n}{2}$. Let us define $\VV^\pr=\{u\in\VV: \exists v\in\VV, v\ne u, \{u,v\}\notin\EE\}$; we have $|\VV^\pr|\le 2d_2$. 
	
	For every strategy profile $y=(y_u)_{u\in\VV^\pr}$, we do the following. For each $v\in\VV\setminus\VV^\pr$, let $n_v^\pr$ be the number of neighbors of $v$ in $\VV^\pr$ who play $1$ in $y$. We now define $g^\pr_v(\el)=g_v(\el+n_v^\pr)$ for every $\el\in\NB\cup\{0\}$ and every player $v\in\VV\setminus\VV^\pr$. For every $k\in\{0,\ldots,|\VV\setminus\VV^\pr|\}$, we compute (using the algorithm in \Cref{obs:complete}) if there exists a PSNE $x^k=(x_v^k)_{v\in\VV\setminus\VV^\pr}$ in the BNPG game $(\GG[\VV\setminus\VV^\pr], (g^\pr_v)_{v\in\VV\setminus\VV^\pr}, (c_v)_{v\in\VV\setminus\VV^\pr})$ where exactly $k$ players play $1$. If $x^k$ exists, then we output \yes if $((y_u)_{u\in\VV^\pr},(x^k_v)_{v\in\VV\setminus\VV^\pr})$ forms a PSNE in the BNPG game $(\GG=(\VV,\EE),(g_v)_{v\in\VV},(c_v)_{v\in\VV})$.
	
	If the above procedure fails to find a PSNE, then we output \no. The running time of the above algorithm is $\OO^*\left(2^{|\VV^\pr|}\right)\leq\OO^*\left(4^{d_2}\right)$. We now argue correctness\longversion{ of our algorithm}.
	
	Clearly, if the algorithm outputs \yes, then there exists a PSNE for the input game. On the other hand, if there exists a PSNE $((y_u)_{u\in\VV^\pr},(x_v)_{v\in\VV\setminus\VV^\pr})\in\{0,1\}^\VV$ in the input game, then let us consider the iteration of our algorithm with the guess $(y_u)_{u\in\VV^\pr}$. Let the number of players playing $1$ in $(x_v)_{v\in\VV\setminus\VV^\pr}$ be $k$. If $x_v=1$ where $v\in\VV\setminus\VV^\pr$, then $\Delta g_v(n_v^\pr+k-1)\geq c_v$ and thus we have $\Delta g^\pr_v(k-1)\geq c_v$. Similarly, if $x_v=0$ where $v\in\VV\setminus\VV^\pr$, then $\Delta g_v(n_v^\pr+k)\leq c_v$ and thus we have $\Delta g^\pr_v(k)\leq c_v$. Hence we observe that $(x_v)_{v\in\VV\setminus\VV^\pr}$ forms a PSNE in the BNPG game $(\GG[\VV\setminus\VV^\pr], (g^\pr_v)_{v\in\VV\setminus\VV^\pr}, (c_v)_{v\in\VV\setminus\VV^\pr})$. Let $(x_v^\pr)_{v\in\VV\setminus\VV^\pr}$ be the PSNE of the BNPG game $(\GG[\VV\setminus\VV^\pr], (g^\pr_v)_{v\in\VV\setminus\VV^\pr}, (c_v)_{v\in\VV\setminus\VV^\pr})$ where exactly $k$ players play $1$ returned by the  algorithm in \Cref{obs:complete}. We observe that every player in $\VV^\pr$ has the same number of neighbors playing $1$ in both the strategy profiles $((y_u)_{u\in\VV^\pr},(x_v)_{v\in\VV\setminus\VV^\pr})$ and $((y_u)_{u\in\VV^\pr},(x_v^\pr)_{v\in\VV\setminus\VV^\pr})$. So no player in $\VV^\pr$ will deviate in the strategy profile $((y_u)_{u\in\VV^\pr},(x_v^\pr)_{v\in\VV\setminus\VV^\pr})$. If $x_v^\pr=1$ where $v\in\VV\setminus\VV^\pr$, then $\Delta g^\pr_v(k-1)\geq c_v$ and thus we have $\Delta g_v(n_v^\pr+k-1)\geq c_v$. Hence, $v$ does not deviate in the strategy profile $((y_u)_{u\in\VV^\pr},(x_v^\pr)_{v\in\VV\setminus\VV^\pr})$. Similarly, if $x_v^\pr=0$ where $v\in\VV\setminus\VV^\pr$, then $\Delta g^\pr_v(k)\leq c_v$ and thus we have $\Delta g_v(n_v^\pr+k)\leq c_v$. Hence, $v$ does not deviate in the strategy profile $((y_u)_{u\in\VV^\pr},(x_v^\pr)_{v\in\VV\setminus\VV^\pr})$. Hence, $((y_u)_{u\in\VV^\pr},(x_v^\pr)_{v\in\VV\setminus\VV^\pr})$ also forms a PSNE in the input BNPG game and thus the algorithm outputs \yes. This concludes the correctness of our algorithm.
\end{proof}

}

%\subsection{Structural Result}

We finally show that a PSNE always exists for fully homogeneous BNPG games for some important graph classes and such a PSNE can be found in $\OO(n)$ time.

\shortversion{
\begin{theorem}[$\star$]\label{thm:rest}
	There is always a PSNE in a fully homogeneous BNPG game for paths, complete graphs, cycles, and bi-cliques. Moreover, we can find a PSNE in $\OO(n)$ time.
\end{theorem}

}

\longversion{
\begin{theorem}\label{thm:rest}
	There is always a PSNE in a fully homogeneous BNPG game for paths, complete graphs, cycles, and bi-cliques. Moreover, we can find a PSNE in $\OO(n)$ time.
\end{theorem}
\begin{proof}
We divide the proof into 4 parts:

\textbf{Part 1- Path:} Let the set of vertices in the input path \PP be $\VV=\{v_1,\ldots,v_n\}$ and the set of edges $\EE=\{\{v_i,v_{i+1}\}:i\in[n-1]\}$. Note that the possible values of $n_v$ for any vertex $v$ in \PP is $0, 1$ and $2$. We show that there is a PSNE in \PP for all possible best response strategies. Let $S_i:=\beta(i)$ (since the game is fully homogeneous, the best-response function is the same for all players) be the set of best responses of a player $v$ if $n_v=i$. Let $x_v$ be the response of a player $v\in \VV$. 

\begin{itemize}
	
	\item If $0\in S_0$, then $(x_v=0)_{v\in\VV}$ forms a PSNE as clearly no player would deviate as $n_v=0$ for every player. So, allow us to assume for the rest of the proof that $S_0=\{1\}$.
	
	\item If $1 \in S_i, \forall i\in\{1,2\}$, then $(x_v=1)_{v\in\VV}$ forms a PSNE as the best response is 1 irrespective of the value of $n_v$ and hence no player would deviate. So, allow us to assume that we have either $1\notin S_1$ or $1\notin S_2$.
	
	\item If $0 \in S_1, 0 \in S_2$, then $((x_{v_i}=0)_{i \equiv1\pmod 2}, (x_{v_i}=1)_{i \equiv0\pmod 2})$ forms a PSNE.  If $i$ is an odd integer then $n_{v_i}>0$ and in this case one of the best responses is 0 and hence $v_i$ does not deviate.  If $i$ is an even integer then $n_{v_i}=0$ and in this case one of the best responses is $1$ (recall $S_0=\{1\}$) and hence $v_i$ does not deviate.
	
	\item If $1 \in S_1, 0 \in S_2$, then $((x_{v_i}=1)_{i \equiv1\pmod 2, i\ne n}, (x_{v_i}=0)_{i \equiv0\pmod 2,i\ne n}, x_{v_n}=1)$ forms a PSNE. If $i$ is an odd integer and not equal to $n$, then $n_{v_i}\leq 1$ and in this case one of the best responses  is 1 and hence $v_i$ does not deviate. If $i$ is an even number and not equal to $n$ then $n_{v_i}=2$ and in this case one of the best responses  is 0 and hence $v_i$ does not deviate. Note that $n_{v_n}\leq 1$ and hence in this case one of the best responses  is 1 and hence $v_n$ does not deviate.
	\item If $0 \in S_1, 1 \in S_2$. In this we have two sub-cases:
	\begin{itemize}
		\item \textbf{n is a multiple of $3$}: In this sub-case, $((x_{v_i}=0)_{i \not\equiv 2\pmod 3}, (x_{v_i}=1)_{i \equiv2\pmod 3})$ forms a PSNE. If we have $i \not\equiv 2 \pmod 3$, then $n_{v_i}=1$ and in this case, $0$ is a best response and hence $v_i$ does not deviate. If $i\equiv2\pmod 3$, then $n_{v_i}=0$ and in this case one of the best responses is $1$ and hence $v_i$ does not deviate. 
		\item \textbf{n is not a multiple of 3}: In this sub-case, $((x_{v_i}=0)_{i \not\equiv 1\pmod 3}, (x_{v_i}=1)_{i \equiv1\pmod 3})$ forms a PSNE.  If $i \not\equiv 1\pmod 3$, then $n_{v_i}=1$ and in this case one of the best responses is $0$ and hence $v_i$ does not deviate. If $i \equiv 1\pmod 3$, then $n_{v_i}=0$ and in this case one of the best responses is 1 and hence $v_i$ does not deviate. 
	\end{itemize}
\end{itemize} 
Since we have a PSNE for every possible best-response function, we conclude that there is always a PSNE in a fully homogeneous BNPG game on paths. Also, we can find a PSNE in paths in $\OO(n)$ time.

\textbf{Part 2- Complete graph:} We assume that the input graph \GG(\VV,\EE) is a complete graph.Let the utility function for all the players $v\in \VV$ be $U(x_v,n_v)=g(x_v+n_v)-c\cdot x_v$. If $\Delta g(n-1)\geq c$, then $(x_v=1)_{v\in \VV}$ is a PSNE. If $\Delta g(0)\leq c$, then $(x_v=0)_{v\in \VV}$ is a PSNE. If $\Delta g(n-1)< c$ and $\Delta g(0)> c$, then there should exist a $0<k\leq n-1$ such that $\Delta g(k)\leq c$ and $\Delta g(k-1)\geq c$ otherwise both $\Delta g(n-1)< c$ and $\Delta g(0)> c$ can't simultaneously hold true. Now we claim that if there exists a $0<k\leq n-1$ such that $\Delta g(k)\leq c$ and $\Delta g(k-1)\geq c$, then choosing any $k$ players and making their response 1 and rest of players response as $0$ would be PSNE. Any player $w$ whose response is $1$ has $n_{w}=k-1$ and since  $\Delta g(k-1)\geq c$, $w$ does not have any incentive to deviate.  Similarly any player $w^\pr$ whose response is $0$ has $n_{w^\pr}=k$ and since  $\Delta g(k)\leq c$, $w^\pr$ does not have any incentive to deviate. This concludes the proof of the theorem as we showed that there is a PSNE in all possible cases. Also clearly we can find a PSNE in \GG in $\OO(n)$ time.

\textbf{Part 3- Cycles:} We assume that the input graph is a Cycle. Let the set of vertices in the input cycle \CC be $\VV=\{v_1,\ldots,v_n\}$ and the set of edges $\EE=\{\{v_i,v_{i+1}\}:i\in[n-1]\}\cup \{v_n,v_1\}$. Note that the possible values of $n_v$(number of neighbours of $v$ choosing 1) for any vertex $v \in \CC$ is $0, 1,$ and $2$. We show that there is a PSNE in \CC for all possible best response strategies. Let $S_i:=\beta(i)$ (since the game is fully homogeneous, the best-response function is the same for all players) denote the set of best responses of a player $v$ if $n_v=i$. Let $x_v$ denote the response of a player $v\in \VV$. 
\begin{itemize}
\item If $0\in S_0$ then $x_v=0$ for every player $v \in \VV$ forms a PSNE as clearly no player would deviate as $n_v=0$ for every player. So, allow us to assume that $S_0=\{1\}$ in the rest of the proof.
\item If $1 \in S_2$, then $x_v=1$ for every player $v$ in \VV forms a PSNE as clearly no player would deviate as $n_v=2$ for every player. So, allow us to assume that $S_2=\{0\}$ in the rest of the proof.
\item If $0 \in S_1$, then $((x_{v_i}=0)_{i \equiv1\pmod 2},(x_{v_i}=1)_{i \equiv0\pmod 2})$ forms a PSNE.  If $i$ is odd then $n_{v_i}>0$ and in this case one of the best responses is $0$ and hence $v_i$ does not deviate.  If $i$ is even number then $n_{v_i}=0$ and in this case one of the best responses is $1$ (recall, we have $S_0=\{1\}$) and hence $v_i$ does not deviate.
\item If $1\in S_1$, then $((x_{v_i}=1)_{i \equiv1\pmod 2}, (x_{v_i}=0)_{i \equiv0\pmod 2})$ forms a PSNE. If $i$ is an odd number then $n_{v_i}\leq 1$ and in this case one of the best responses  is $1$ and hence $v_i$ does not deviate. If $i$ is an even number and then $n_{v_i}=2$ and in this case one of the best responses  is 0 and hence $v_i$ does not deviate.
\end{itemize} 
Since we have a PSNE for every possible best-response functions, we conclude that there is always a PSNE in a fully homogeneous BNPG game on cycles. Also, we can find a PSNE in cycles in $\OO(n)$ time.

\textbf{Part 4- Bicliques:} Let the input graph $\GG=(\VV,\EE)$ be a biclique; \VV is partitioned into $2$ sets namely $\VV_1=\{u_1,\ldots,u_{n_1}\}$ and $\VV_2=\{v_1,\ldots,v_{n_2}\}$ where $n_1+n_2=n$ and $\EE=\{(u_i,v_j):i\in[n_1],j\in[n_2]\}$. We show that there is a PSNE in \CC for all possible best response strategies. Let $S_i:=\beta(i)$ (since the game is fully homogeneous, the best-response function is the same for all players) denote the set of best responses of a player $v$ if $n_v=i$. Let $x_v$ denote the response of a player $v\in \VV$. 
\begin{enumerate}
	
\item \textbf{If $n_1=n_2$}: For this case we have the following sub-cases: 

\begin{itemize}
\item If $0\in S_0$, then $(x_v=0)_{v\in\VV}$ forms a PSNE as clearly no player would deviate as $n_v=0$ for every player. So, allow us to assume that $S_0=\{1\}$.
\item If $1\in S_{n_1}$, then $(x_v=1)_{v\in\VV}$ forms a PSNE as clearly no player would deviate as $n_v=n_1$ for every player. So, allow us to assume that $S_{n_1}=\{0\}$. However, then $((x_{v}=0)_{v\in\VV_1}, (x_v=1)_{v\in\VV_2})$ forms a PSNE.  If $v\in \VV_2$ then $n_{v}=0$ and in this case one of the best responses is 1 and hence $v$ won't deviate. If $v\in \VV_1$ then $n_{v}=n_1$ and in this case one of the best responses is 0 and hence $v$ won't deviate.
\end{itemize}

\item \textbf{If $n_1\ne n_2$}: For this case we have the following sub-cases:

\begin{itemize}
\item If $0\in S_0$ then $x_v=0$ for every player $v$ in \VV forms a PSNE as clearly no player would deviate as $n_v=0$ for every player. So, allow us to assume that $S_0=\{1\}$ for the rest of the proof.

\item If $1\in S_{n_1}$, $1\in S_{n_2}$, then $(x_v=1)_{v\in\VV}$ forms a PSNE. If $v\in \VV_1$, then $n_{v}=n_2$ and in this case one of the best responses is $1$ and hence $v$ does not deviate. If $v\in \VV_2$, then $n_{v}=n_1$ and in this case one of the best responses is 1 and hence $v$ does not deviate.

\item  If $0\in S_{n_1},1\in S_{n_2} $, then $((x_v=1)_{:v\in \VV_1}, (x_v=0)_{v\in \VV_2})$ forms a PSNE. If $v\in \VV_1$ then $n_{v}=0$ and in this case one of the best responses is $1$ (recall $S_0=\{1\}$) and hence $v$ does not deviate. If $v\in \VV_2$ then $n_{v}=n_1$ and in this case one of the best responses is $0$ and hence $v$ does not deviate.

\item  If $0\in S_{n_1},0\in S_{n_2} $, then $((x_v=1)_{v\in \VV_1}, (x_v=0)_{v\in \VV_2})$ forms a PSNE. If $v\in \VV_1$ then $n_{v}=0$ and in this case one of the best responses is $1$ (recall $S_0=\{1\}$) and hence $v$ does not deviate. If $v\in \VV_2$ then $n_{v}=n_1$ and in this case one of the best responses is 0 and hence $v$ does not deviate.

\item  If $1\in S_{n_1},0\in S_{n_2} $, then $((x_v=0)_{:v\in \VV_1}, (x_v=1)_{v\in \VV_2})$ forms a PSNE. If $v\in \VV_1$ then $n_{v}=n_2$ and in this case one of the best responses is $0$ and hence $v$ does not deviate. If $v\in \VV_2$ then $n_{v}=0$ and in this case one of the best responses is $1$ (recall $S_0=\{1\}$) and hence $v$ does not deviate.
\end{itemize}
\end{enumerate}
Since we have a PSNE for every possible best-response functions, we conclude that there is always a PSNE in a fully homogeneous BNPG game on biclique. Also, we can find a PSNE in biclique in $\OO(n)$ time.
\end{proof}

}

\section{Conclusion and Future Work}

We have studied parameterized complexity of the \psne problem for the BNPG games with respect to various important graph parameters. We exhibited intractibility w.r.t. the parameters like maximum degree, diameter, treedepth, number of players playing $1$ and $0$. We complemented this by showing \FPT algorithms parameterized by circuit rank, treewidth+maximum degree, and the distance from complete graph. We also showed that PSNE always exists in a fully homogeneous BNPG game for paths, complete graphs, cycles and bi-cliques.

Our work leaves some important questions open. For example, can we show PPAD-Hardness for finding Nash Equilibrium in BNPG games. Another immediate research direction is to study if our algorithmic results could be extended to other types of more general public goods games. Another research direction could be to look at social welfare functions in the context of BNPG game. We can also consider BNPG games with altruism introduced in \cite{yu2021altruism} and try to resolve its parameterized complexity.
\bibliography{references}
\newpage
\appendix
%\section{Standard Definitions}\label{std:def}

\section{Missing Proofs}
\begin{proof}[Proof of \Cref{lem:best-util}]
Consider an arbitrary player $v$. Let the best response function of $v$ be $\beta_v$. We now define the utility function of $v$ as $U_{v}(x_v,n_v)=g_v(x_v+n_v)-c_v\cdot x_v$ where $c_v$ is a constant greater than 1 and $g_v(0)$ is a constant greater than 0. Now we define $g_v(.)$ recursively in the following way:\\
 $g_v(x)= \left\{ \begin{array}{rcl}
g_v(x-1)+c_v-1 & \mbox{if} & \beta_v(x-1)=\{0\}  \\ 
g_v(x-1)+c_v+1 & \mbox{if} & \beta_v(x-1)=\{1\}  \\ 
g_v(x-1)+c_v & \mbox{if} & \beta_v(x-1)=\{0,1\}  \\ 
\end{array}\right.$\\
In the above recursive definition, $x$ belongs to $[n-1]$. Let $x^\pr$ be an arbitrary number in $[n-1]$. From the recursive definition, we can conclude that $1\in\beta_v(x^\pr-1)$ iff $g_v(x^\pr)-c_v\geq g_v(x^\pr-1)$. Similarly, we can conclude that $0\in\beta_v(x^\pr-1)$ iff $g_v(x^\pr)-c_v\leq g_v(x^\pr-1)$. Since $x^\pr$ and $v$ were chosen arbitrarily, we can conclude that 
 for every player $w$, for every $k\in\{0,1,\ldots,n-1\}$ and for every $a\in\{0,1\}$, we have $a\in\beta_w(k)$ if and only if, for every strategy profile $x_{-w}$ of players other than $w$ where exactly $k$ players in the neighborhood of $w$ play $1$, we have $U_w(x_w=a,x_{-w}) \ge U_w(x_w=a^\pr,x_{-w})$ for all $a^\pr\in\{0,1\}$.
\end{proof}
\begin{proof}[Proof of \Cref{thm:k0}]
	Let $d(v)$ denote the degree of $v$ in \GG. To prove the result for the parameter $k_0$, we use the following best-response function.
	\[ \beta_v(k^\pr) = \begin{cases}
		0 & \text{if } k^\pr=d(v)\\
		\{0,1\} & \text{otherwise}
	\end{cases} \]
	
	We claim that the above BNPG game has a PSNE having at most $k$ players playing $0$ if and only if the \ds instance is a \yes instance.
	
	For the ``if'' part, suppose the \ds instance is a \yes instance and $\WW\subseteq\VV$ be a dominating set for \GG of size at most $k$. We claim that the strategy profile $\bar{x}=((x_v=0)_{v\in\WW}, (x_v=1)_{v\in\VV\setminus\WW})$ is a PSNE for the BNPG game. To see this, we observe that every player $w \in \VV\setminus\WW$ has at least $1$ neighbor playing $0$ , and thus she has no incentive to deviate as $n_w<d(v)$. On the other hand, since $0$ is always a best-response strategy for every player irrespective of what others play, the players in \WW also do not have any incentive to deviate. Hence, $\bar{x}$ is a PSNE.
	
	For the ``only if'' part, let $\bar{x}=((x_v=0)_{v\in\WW}, (x_v=1)_{v\in\VV\setminus\WW})$ be a PSNE for the BNPG game where $|\WW|\le k$ (that is, at most $k$ players are playing $0$). We claim that \WW forms a dominating set for \GG. Indeed this claim has to be correct, otherwise there exists a vertex $w\in\VV\setminus\WW$ which does not have any neighbor in \WW and consequently, the player $w$ has incentive to deviate to $0$ from $1$ as $n_w=d(v)$ which is a contradiction.
\end{proof}

\begin{proof}[Proof of \Cref{thm:k1}]
	Let $(\GG=(\VV,\EE),k)$ an arbitrary instance of \ds. We consider a fully homogeneous BNPG game on the same graph \GG whose best-response functions $\beta_v(\cdot)$ for $v\in\VV$ is given below:
	\[ \beta_v(k^\pr) = \begin{cases}
		1 & \text{if } k^\pr=0\\
		\{0,1\} & \text{otherwise}
	\end{cases} \]
	
	We claim that the above BNPG game has a PSNE having at most $k$ players playing $1$ if and only if the \ds instance is a \yes instance.
	
	For the ``if'' part, suppose the \ds instance is a \yes instance and $\WW\subseteq\VV$ be a dominating set for \GG of size at most $k$. We claim that the strategy profile $\bar{x}=((x_v=1)_{v\in\WW}, (x_v=0)_{v\in\VV\setminus\WW})$ is a PSNE for the BNPG game. To see this, we observe that every player in $\VV\setminus\WW$ has at least $1$ neighbor playing $1$ , and thus she has no incentive to deviate. On the other hand, since $1$ is always a best-response strategy for every player irrespective of what others play, the players in \WW also do not have any incentive to deviate. Hence, $\bar{x}$ is a PSNE.
	
	For the ``only if'' part, let $\bar{x}=((x_v=1)_{v\in\WW}, (x_v=0)_{v\in\VV\setminus\WW})$ be a PSNE for the BNPG game where $|\WW|\le k$ (that is, at most $k$ players are playing $1$). We claim that \WW forms a dominating set for \GG. Indeed, this claim has to be correct, otherwise there exists a vertex $w\in\VV\setminus\WW$ which does not have any neighbor in \WW and consequently, the player $w$ has incentive to deviate to $1$ from $0$ as $n_w=0$ which is a contradiction.
\end{proof}

\begin{proof}[Proof of \Cref{thm:fully-deg}]
	The high-level idea in this proof is the same as the proof of \Cref{thm:full-homo}. The only difference is that we add the special nodes in a way such that the maximum degree in the instance of fully homogeneous BNPG game is upper bounded by $9$.
	
	Formally, we consider an instance of a heterogeneous BNPG game on a graph $\GG=(\VV=\{v_i: i\in[n]\},\EE)$ such that there are only $2$ types of utility functions $U_1(x_v,n_v)=g_1(x_v+n_v)-c_1x_v, U_2(x_v,n_v)=g_2(x_v+n_v)-c_2x_v,$ and the degree of any vertex is at most $3$; we know from \Cref{thm:deg} that it is an \NPC instance. Let us partition $\VV$ into $\VV_1$ and $\VV_2$ such that the utility function of the players in $\VV_1$ is $U_1(\cdot)$ and the utility function of the players in $\VV_2$ is $U_2(\cdot)$. We now construct the graph $\HH=(\VV^\pr,\EE^\pr)$ for the instance of the fully homogeneous BNPG game. 
	\begin{align*}
		\VV^\pr &= \{w_i: i\in[n]\}\cup\WW_1\cup\WW_2, \text{where}\\
		\WW_i &= \{a_j^k: j\in[2+4(i-1)], v_k\in\VV_i\} \;\forall i\in[2]\\
		\EE^\pr &= \{\{w_i,w_j\}: \{v_i,v_j\}\in\EE\}\cup\EE_1\cup\EE_2, \text{where}\\
		\EE_i &= \{\{a_j^k,w_k\}: j\in[2+4(i-1)], v_k\in\VV_i\} \;\forall i\in[2]
	\end{align*}
	
	We define two functions --- $f(x)=\lfloor\frac{x-2}{4}\rfloor+1$ and $h(x)=x-2-4(f(x)-1)$. We now define best-response functions for the players in \HH.
	\[
	\beta(k)=\begin{cases}
		1 & \text{if } k=0 \text{ or } k=1\\
		\{0,1\} & \Delta g_{f(k)}(h(k))=c_{f(k)},k>1 \\
		1 & \Delta g_{f(k)}(h(k))>c_{f(k)},k>1 \\
		0 & \Delta g_{f(k)}(h(k))<c_{f(k)},k>1
	\end{cases}
	\]
	
	This finishes description of our fully homogeneous BNPG game on \HH. We now claim that there exists a PSNE in the heterogeneous BNPG game on \GG if and only if there exists a PSNE in the fully homogeneous BNPG game on \HH. We note that degree of any node in \HH is at most $9$.
	
	For the ``only if'' part, let $x^*=(x^*_v)_{v\in \VV}$ be a PSNE in the heterogeneous BNPG game on \GG. We now consider the following strategy profile $\bar{y}=(y_v)_{v\in\VV^\pr}$ for players in \HH.
	\[
	\forall i\in[n] y_{w_i} = x^*_{v_i}; y_b=1 \text{ for other vertices }b
	\]
	
	We now claim that the players in $\VV'$ also does not deviate. Clearly the players in $\cup_{i\in[2]} \WW_i$ do not deviate as their degree is $1$ and $\beta(0)=\beta(1)=1$. If $v_k\in\VV_i$, then in $\bar{y}$ we have $n_{w_k} = n_{v_k}+2+4(i-1)\ge 2$ for every $i\in[2]$ and $n_{v_k}\le 3$ as the maximum degree in \GG is $3$. If $x_{v_k}^*=1$, then we have $\Delta g_{i}(n_{v_k})\geq c_{i}$. We have $f(n_{w_k})=i$ and $h(n_{w_k})=n_{v_k}$. This implies that $\Delta g_{f(n_{w_k})}(h(n_{w_k}))\geq c_{f(n_{w_k})}$. So $w_k$ does not deviate as $1$ is the best-response. If $x_{v_k}^*=0$, then we have $\Delta g_{i}(n_{v_k})\leq c_{i}$. This implies that $\Delta g_{f(n_{w_k})}(h(n_{w_k}))\leq c_{f(n_{w_k})}$. So $w_k$ does not deviate as $0$ is the best-response. Hence, $\bar{y}$ is a PSNE.
	
	For the ``if'' part, suppose there exists a PSNE $(x^*_v)_{v\in \VV^\pr}$ in the fully homogeneous BNPG game on \HH. Clearly $x^*_w=1$ for all $w\in \cup_{i\in[2]} \WW_i$ as $n_w\leq 1$. Now we claim that the strategy profile $\bar{x}=(x_{v_k}=x_{w_k}^*)_{k\in[n]}$ forms a PSNE for the heterogeneous BNPG game on \GG. We observe that if $x_{w_k}^*=1$ and $v_k\in\VV_i$, then $\Delta g_{f(n_{w_k})}(h(n_{w_k}))\geq c_{f(n_{w_k})}$ for $k\in[n]$. This implies that $\Delta g_{i}(n_{v_k})\geq c_{i}$.  So $x_{v_k}=1$ is the best-response for $v_k\in\VV_i$ and hence, she does not deviate. Similarly, If $x_{w_k}^*=0$ and $v_k\in\VV_i$, then $\Delta g_{f(n_{w_k})}(h(n_{w_k}))\leq c_{f(n_{w_k})}$. This implies that $\Delta g_{i}(n_{v_k})\leq c_{i}$.  So $x_{v_k}=0$ is the best-response for $v_k\in \VV_i$ and hence, it won't deviate. Hence, $\bar{x}$ is a PSNE in the heterogeneous BNPG game on \GG.
\end{proof}

\begin{proof}[Proof of \Cref{feasible:D}]
	We use dynamic programming to solve this problem. Let $N(V)$ denote set of vertices in \VV which is adjacent to at least one vertex in $V$. Let $N[V]:= V\cup N(V)$. Let $N[V]=\{u_1,\ldots,u_l\}$ where $l$ is atmost $O(|V|\cdot \Delta)$. Let $c[(d_u)_{u\in V},i]$ denote whether there exists a strategy profile $S$ such that for each $u\in V$, number of neighbours of $u$ in the set $\{u_1,\ldots,u_i\}$ ($\phi$ if $i=0$) playing $1$ in the strategy profile $S$ is $d_u$. Let $f$ be a function such that $f(u)=d_u$ for all $u\in V$. Clearly, $c[(d_u)_{u\in V},l]$ indicates whether the function $f$ is feasible or not. Let $g:\VV\times \VV\rightarrow \{0,1\}$ be a function such that $g(\{u,v\})=1$ if and only if $\{u,v\}\in \EE$. We now present the recursive equation to compute $c[(d_u)_{u\in V},i]$:
	\[      
	c[(d_u)_{u\in V},i]=\begin{cases}
		0 \text{ if }\exists u,d_u<0\\
		c[(d_{u}- g(\{u,u_i\}))_{u\in V},i-1] \vee c[(d_u)_{u\in V},i-1] \text{ if } i\geq1 \\
		1  \text{ if $i=0$ and }\forall u\in V,d_u=0 \\
		0 \text{ otherwise }
	\end{cases}
	\]
	
	Now we argue for the correctness of the above recursive equation. Base cases are trivial as it follows from the definitions. We now look at the case where $i\geq 1$ and $\forall u\in V, d_u\geq 0$. If $c[(d_{u}- g(\{u,u_i\}))_{u\in V},i-1]=1$, then consider the strategy profile $S$ which makes it $1$. Now consider a strategy profile where the response of $u_i$ is $1$ and rest of the players play as per $S$. So for any $u\in V$, the number of neighbours in $\{u_1,\ldots,u_i\}$ playing $1$ increases by $g(\{u,u_i\})$ when compared to the number of neighbours in $\{u_1,\ldots,u_{i-1}\}$ playing $1$. This would imply that $c[(d_u)_{u\in V},i]=1$. Similarly, if $ c[(d_u)_{u\in V},i-1]=1$ ,  then considering the response of $u_i$ as $0$ will not change the number of neighbours playing $1$ and therefore $c[(d_u)_{u\in V},i]=1$. In the other direction, if $c[(d_u)_{u\in V},i]=1$, then consider the strategy profile $S$ which makes it $1$. If the response of $u_i$ is $1$ (resp. $0$) in $S$, then the number of neighbours of $u$ in $\{u_1,\ldots,u_{i-1}\}$ playing $1$ decreases by $g(\{u,u_i\}$ (resp. $0$) when compared with the number of neighbours in $\{u_1,\ldots,u_{i}\}$ playing $1$. Hence, $c[(d_{u}- g(\{u,u_i\}))_{u\in V},i-1] \vee c[(d_u)_{u\in V},i-1]$ is equal to $1$.
	
	Now we look at the time complexity. Total number of cells in the dynamic programming table which we created is $\OO^*(\Delta^{|V|})$ as the value of each entry in $(d_u)_{u\in V}$ is at most $\Delta$. Time spent in each cell is $n^{O(1)}$. Hence, the set of all feasible functions $f:V\rightarrow [\Delta]\cup\{0\}$ can be computed in time $\OO^*(\Delta^{|V|})$.
\end{proof}

\begin{proof}[Proof of \Cref{thm:XP}]
	Let $(\GG=(\VV,\EE),(g_v)_{v\in\VV},(c_v)_{v\in\VV})$ be any instance of \psne for BNPG games. Let $(\beta_v(.))_{v\in\VV}$ be the set of the best response functions. Let $\TT = (T, \{X_t \}_{t\in V (T )} )$ be a nice tree decomposition of the input
	$n$-vertex graph $\GG$ that has width at most $k$. Let \TT be rooted at some node $r$. For a node $t$ of \TT , let $V_t$ be the union of all the bags present in the subtree of \TT rooted at $t$, including $X_t$. We solve the \psne problem using dynamic programming. Let $N_1(X_t)$ denote set of vertices in $\VV\setminus V_t$ which is adjacent to at least one vertex in $X_t$. Let $N_2(X_t)$ denote set of vertices in $V_t\setminus X_t$ which is adjacent to at least one vertex in $X_t$.Let $c[t,(x_v)_{v\in{X_t}},(d_v^1)_{v\in{X_t}},(d_v^2)_{v\in{X_t}}] =1$ (resp. 0) denote that there exists (resp. doesn't exist) a strategy profile $S$ of all the players in \GG such that for each $u\in X_t$, $u$ plays $x_u$, number of neighbours of $u$ in $N_1(X_t)$ (resp. $N_2(X_t)$) playing $1$ is $d_u^1$ (resp. $d_u^2$) and none of the vertices in $V_t$ deviate in the strategy profile $S$. Before we proceed, we would like to introduce some notations. Let $V$ be a set of vertices and $S_1=(x_v)_{v\in V}$, $S_2=(x_v)_{v\in V\setminus\{w\}}$ be two tuples. Then $S_1\setminus\{x_w\}:=S_2$ and $S_2\cup\{x_w\}:=S_1$. Also, we denote an empty tuple by $\phi$. Clearly $c[r,\phi,\phi,\phi]$ indicates whether there is a PSNE in \GG or not. We now present the recursive equation to compute $c[t,(x_v)_{v\in{X_t}},(d_v^1)_{v\in{X_t}},(d_v^2)_{v\in{X_t}}]$ for various types of node in $\TT$.
	
	\textbf{Leaf Node:} For a leaf node $t$ we have that $X_t = \phi$. Hence, $c[t,\phi,\phi,\phi]=1$.
	
	\textbf{Join Node:} For a join node $t$, let $t_1,t_2$ be its two children. Note that $X_t=X_{t_1}=X_{t_2}$. Now we proceed to compute $c[t,(x_v)_{v\in{X_t}},(d_v^1)_{v\in{X_t}},(d_v^2)_{v\in{X_t}}]$. %If there is a vertex $v\in V_t$ such that $x_v\notin \beta(n_v'+d_v^1+d_v^2)$ where $n_v'$ is the number of neighbours of $v$ in $X_t$ playing  $1$ in $(x_v)_{v\in{X_t}}$, then $c[t,(x_v)_{v\in{X_t}},(d_v^1)_{v\in{X_t}},(d_v^2)_{v\in{X_t}}]=0$. 
	Let $\mathcal{F}$ be a set of tuples $(d_v')_{v\in X_t}$ such that there is a strategy profile $S$ such that for each $v\in X_t$, its response is $x_v$, the number of neighbours in $N_1(x)$, $V_{t_1}\setminus X_{t_1}$ and $V_{t_2}\setminus X_{t_2}$ playing $1$ is $d_v^1,d_v',d_v^2-d_v'$ respectively. Using \Cref{feasible:D} we can find the set $\mathcal{F}$ in time $\OO^*(\Delta^k)$.Then $c[t,(x_v)_{v\in{X_t}},(d_v^1)_{v\in{X_t}},(d_v^2)_{v\in{X_t}}]$ is equal to the following formula:
\begin{align*}
&0\vee\bigvee_{(d_v')_{v\in X_t}\in \mathcal{F}}\big(c[t_1,(x_v)_{v\in{X_t}},(d_v^1+d_v^2-d_v')_{v\in{X_t}},(d_v')_{v\in{X_t}}]\\
&\wedge c[t_2,(x_v)_{v\in{X_t}},(d_v^1+d_v')_{v\in{X_t}},(d_v^2-d_v')_{v\in{X_t}}]\big)
\end{align*}

	Now we argue for the correctness of the above recursive equation. If $\mathcal{F}$ is empty then the above equation trivially holds true. Hence, we assume that $\mathcal{F}$ is non-empty. In one direction let us assume that there exists a tuple $(d_v')_{v\in X_t}\in \mathcal{F}$ such that $(c[t_1,(x_v)_{v\in{X_t}},(d_v^1+d_v^2-d_v')_{v\in{X_t}},(d_v')_{v\in{X_t}}]\wedge c[t_2,(x_v)_{v\in{X_t}},(d_v^1+d_v')_{v\in{X_t}},(d_v^2$$-d_v')_{v\in{X_t}}])$=1. Let $S_1$ and $S_2$ be the strategy profiles which leads to $c[t_1,(x_v)_{v\in{X_t}},(d_v^1+d_v^2-d_v')_{v\in{X_t}},(d_v')_{v\in{X_t}}]$ and $c[t_2,(x_v)_{v\in{X_t}},(d_v^1+d_v')_{v\in{X_t}},(d_v^2$$-d_v')_{v\in{X_t}}]$ respectively being $1$. Let $S_3$ be a strategy profile that leads to $(d_v')_{v\in X_t}$ being included in $\mathcal{F}$. Now consider a strategy profile $S$ such that responses of players in $V_{t_1}\setminus X_{t}$ is their responses in $S_1$, responses of the players in $V_{t_2}\setminus X_{t}$ is their responses in $S_2$ and responses of rest of the players is their responses in $S_3$. Now observe that there are no edges between the set of vertices $V_{t_1}\setminus X_t$ and $V_{t_2}\setminus X_t$ and there are no edges between the set of vertices $\VV\setminus V_t$ and $V_{t}\setminus X_t$. Therefore, the number of neighbours of vertices in $V_{t_1}\setminus X_{t_1}$ and $V_{t_2}\setminus X_{t_2}$ playing $1$ doesn't change when compared to $S_1$ and $S_2$ respectively. Also, the number of neighbours of vertices in $X_t$ playing $1$ doesn't change when compared to $S_1$, $S_2$ and $S_3$. Hence, $c[t,(x_v)_{v\in{X_t}},(d_v^1)_{v\in{X_t}},(d_v^2)_{v\in{X_t}}]=1$. In other direction, let $c[t,(x_v)_{v\in{X_t}},(d_v^1)_{v\in{X_t}},(d_v^2)_{v\in{X_t}}]=1$. Let $S'$ be the strategy profile which leads to $c[t,(x_v)_{v\in{X_t}},(d_v^1)_{v\in{X_t}},(d_v^2)_{v\in{X_t}}]$ being $1$. For each $v\in X_t$, let the number of neighbours in $V_{t_1}\setminus X_t$ playing $1$ in $S'$ be $d_v'$. Then clearly $S'$ leads to both $c[t_1,(x_v)_{v\in{X_t}},(d_v^1+d_v^2-d_v')_{v\in{X_t}},(d_v')_{v\in{X_t}}]$ and $c[t_2,(x_v)_{v\in{X_t}},(d_v^1+d_v')_{v\in{X_t}},(d_v^2$$-d_v')_{v\in{X_t}}]$ being $1$.
	
	\textbf{Introduce Node:} Let $t$ be an introduce node with a child $t'$ such that $X_t = X_{t'} \cup \{u\}$ for some $u \notin X_{t'}$. Let $S'=(x_v)_{v\in{X_t}}$ be a strategy profile of vertices in $X_t$. Let $n_v'$ denote the number of neighbours of $v$ playing $1$ in $S'$.  Let $g:\VV\times \VV\rightarrow \{0,1\}$ be a function such that $g(\{u,v\})=1$ if and only if $\{u,v\}\in \EE$.  We now proceed to compute $c[t,S',(d_v^1)_{v\in{X_t}},(d_v^2)_{v\in{X_t}}]$. If there is no strategy profile $S$ where $\forall v\in X_t$,the number of neighbours of $v$ in $N_1(X_t)$ (resp. $N_2(X_t)$) playing $1$ is $d_v^1$ (resp. $d_v^2$) , then clearly $c[t,S',(d_v^1)_{v\in{X_t}},(d_v^2)_{v\in{X_t}}]=0$. Due to \Cref{feasible:D}, we can check the previous statement in $\OO^*(\Delta^{k})$ by considering a bipartite subgraph of \GG between $X_t$ and $N_1(X_t)$ (or $N_2(X_t)$).  Otherwise, we have the following:

	\[	
	c[t,S',(d_v^1)_{v\in{X_t}},(d_v^2)_{v\in{X_t}}]=
	\begin{cases}
		0 \text{ if }\exists v\in X_t, x_v\notin\beta_v(n_v'+d_v^1+d_v^2)\\
		c[t',S'\setminus \{x_u\},(d_v^1+g(\{v,u\}))_{v\in{X_{t'}}},(d_v^2)_{v\in{X_{t'}}}] \text{ if }x_u=1\\
		c[t',S'\setminus \{x_u\},(d_v^1)_{v\in{X_{t'}}},(d_v^2)_{v\in{X_{t'}}}] \text{ otherwise }
	\end{cases}
	\]
	Now we argue for the correctness of the above recursive equation. The base case is trivial. Now if $x_u=1$ and $c[t',S'\setminus \{x_u\},(d_v^1+g(\{v,u\}))_{v\in{X_{t'}}},(d_v^2)_{v\in{X_{t'}}}]=1$, then let $S_1$ be the strategy profile that makes the cell value to be $1$.
	Let $S_2$ be a strategy such that $\forall v\in X_t$, the number of neighbours of $v$ in $N_1(X_t)$ (resp. $N_2(X_t)$) playing $1$ is $d_v^1$ (resp. $d_v^2$). Let $S_3$ be a strategy profile where the responses of players in $X_t$ is $(x_v)_{v\in X_t}$, the responses of players in $V_t\setminus X_t$ is their responses in $S_1$ and the responses of the rest of the players is their responses in $S_2$. Now observe that there is no edge between the set of vertices $\VV\setminus V_{t'}$ and $V_{t'}\setminus X_{t'}$ (due to this the only valid value of $d_u^2$ is 0 which is ensured by us). Therefore, none of the vertices deviate in $V_{t'}\setminus X_{t'}$ deviate. Also, the number of neighbours of the vertices in $X_{t'}$ playing $1$ in $S_3$ is the same as that of $S_1$. Hence, the vertices in $X_{t'}$ don't deviate. Vertex $u$ doesn't deviate as we are not in the base case and therefore $x_u\in\beta_v(n_u'+d_u^1+d_u^2)$. Hence, $c[t,S',(d_v^1)_{v\in{X_t}},(d_v^2)_{v\in{X_t}}]=1$. By similar arguments we can show that if $x_u=0$ and $c[t',S'\setminus \{x_u\},(d_v^1))_{v\in{X_{t'}}},(d_v^2)_{v\in{X_{t'}}}]=1$ then $c[t,S',(d_v^1)_{v\in{X_t}},(d_v^2)_{v\in{X_t}}]=1$. In the other direction, if $c[t,S',(d_v^1)_{v\in{X_t}},(d_v^2)_{v\in{X_t}}]=1$, then let $S_4$ be the strategy profile which makes the value of this cell $1$. If $x_u=1$ in $S_4$, then for each $v\in X_{t'}$, the number of neighbours of $X_{t'}$ in $N_1(X_{t'})$ (resp. $N_2(X_{t'}))$ playing $1$ in $S_4$ is $d_v^1+g(\{v,u\})$ (resp. $d_v^2$). Since none of the vertices in $V_t$ deviate, so clearly $S_4$ leads $c[t',S'\setminus \{x_u\},(d_v^1+g(\{v,u\}))_{v\in{X_{t'}}},(d_v^2)_{v\in{X_{t'}}}]$ to being $1$ . Similar argument holds when $x_u=0$ in $S_4$ , and it would lead $c[t',S'\setminus \{x_u\},(d_v^1)_{v\in{X_{t'}}},(d_v^2)_{v\in{X_{t'}}}]$ to being $1$.
	
	\textbf{Forget Node:} Let $t$ be a forget node with a child $t'$ such that $X_t$ = $X_t' \setminus \{w\}$ for some $w \in X_{t'}$. Let $S_0=(x_v)_{v\in{X_t}}\cup \{x_w=0\},$ $S_1=(x_v)_{v\in{X_t}}\cup \{x_w=1\}$ be two strategy profiles of vertices in $X_t'$. Let $g:\VV\times \VV\rightarrow \{0,1\}$ be a function such that $g(\{u,v\})=1$ if and only if $\{u,v\}\in \EE$. We now have the following:
	\begin{align*}
c[t,(x_v)_{v\in X_t},(d_v^1)_{v\in{X_t}},(d_v^2)_{v\in{X_t}}]=&\bigvee_{d_w^1,d_w^2:0\leq d_w^1,d_w^2\leq \Delta}\big(c[t',S_0,(d_v^1)_{v\in{X_{t'}}},(d_v^2)_{v\in{X_{t'}}}]\\
&\vee c[t',S_1,(d_v^1)_{v\in{X_{t'}}},(d_v^2-g(\{v,w\}))_{v\in{X_{t'}}}\big)   	
\end{align*}

	Now we argue for the correctness of the above recursive equation. In one direction, let us assume that $\exists d_w^1,d_w^2\in [\Delta]\cup \{0\}$ such that $(c[t',S_0,(d_v^1)_{v\in{X_{t'}}},(d_v^2)_{v\in{X_{t'}}}]\vee c[t',S_1,(d_v^1)_{v\in{X_{t'}}},(d_v^2-g(\{v,w\}))_{v\in{X_{t'}}}])=1$. Let $S$ be the strategy profile that made the formula in the previous statement to be true. Now observe that for each $v\in X_t$, the number of neighbours in $N_1(X_t)$ and $N_2(X_t)$ is $d^1_v$ and $d^2_v$ respectively. Since none of the vertices in the set $V_t=V_{t'}$ deviate, therefore $c[t,(x_v)_{v\in{N[X_t]}},(d_v^1)_{v\in{X_t}},(d_v^2)_{v\in{X_t}}]=1$. In other direction, let us assume that $c[t,(x_v)_{v\in{N[X_t]}},(d_v^1)_{v\in{X_t}},(d_v^2)_{v\in{X_t}}]=1$. Let $S'$ be the strategy profile that made the formula in the previous statement to be true. Let $d_w^1$ and $d_w^2$ be the number of neighbours of $w$ in $N_1(X_{t'})$ and $N_2(X_{t'})$ respectively playing $1$ in the profile $S'$. For each $v\in X_{t}$, if $x_w=1$ (resp. $x_w=0$) then the number of neighbours in $N_2(X_{t'})$ playing $1$ in $S'$ is $d_v^2-g(\{v,w\}$ (resp. $d_v^2$).  Similarly, for each $v\in X_{t}$, the number of neighbours in $N_1(X_{t'})$ playing $1$ in $S'$ is $d_v^1$. Since none of the vertices in $V_t=V_{t'}$ deviate in $S'$, therefore   $(c[t',S_0,(d_v^1)_{v\in{X_{t'}}},(d_v^2)_{v\in{X_{t'}}}]\vee c[t',S_1,(d_v^1)_{v\in{X_{t'}}},(d_v^2-g(\{v,w\}))_{v\in{X_{t'}}}])=1$.
	
	Now we look at the time complexity. Total number of cells in the dynamic programming table which we created is $\OO^*(\Delta^{O(k)})$. For each cell, we spend at most $\OO^*(\Delta^{O(k)})$ time if we are computing the table in a bottom up fashion. Hence, the running time is $\OO^*(\Delta^{O(k)})$. 
\end{proof}

\begin{proof}[Proof of \Cref{vertex-cover1}]
	Let $(\GG=(\VV,\EE),(g_v)_{v\in\VV},(c_v)_{v\in\VV})$ be any instance of \psne for BNPG games. We compute a minimum vertex cover $\SS\subset\VV$ in time $\OO^*(2^{\vc})$~\cite{CyganFKLMPPS15}. The idea is to brute force on the strategy profile of players in \SS and assign actions of other players based on their best-response functions. For every strategy profile $x_\SS=(x_v)_{v\in\SS}$, we do the following.
	
	\begin{enumerate}
		\item For $w\in\VV\setminus\SS$, let $n_w$ be the number of neighbors of $w$ (they can only be in \SS) who play $1$. We define $x_w = 1$ if  $\Delta g_{w}(n_w)> c_{w}$ and $0$ if $\Delta g_{w}(n_w)< c_{w}$. This is well-defined since $\Delta g_{w}(n_w)\ne c_{w}$ as the game is strict.
		\item If $(x_v)_{v\in \VV}$ forms a PSNE, then output \yes. Otherwise, we discard the current $x_\SS$.
	\end{enumerate}
	If the above the procedure does not output \yes for any $x_\SS$, then we output \no. The correctness of the algorithm is immediate. Since the computation for every guess of $x_\SS$ can be done in polynomial time and the number of such guesses is $2^\SS=2^\vc$, it follows that the running time of our algorithm is $\OO^*(2^\vc)$.
\end{proof}

\begin{proof}[Proof of \Cref{obs:add}]
	$\forall x\geq 0,\forall i\in[n]$, $g_{v_i}(x+1)-g_{v_i}(x)=g_{v_i}(1)$. This implies for a player $v_i$, the best response doesn't depend on the responses of its neighbours and solely depends on $g_{v_i}(1)$. Hence, if $g_{v_i}(1)\geq c_{v_i}$ then we assign the response of player $v_i$ as 1 and 0 otherwise. This will make sure that no player $v_i$ deviates. So calculating the PSNE takes $\OO(n)$
\end{proof}

\begin{proof}[Proof of \Cref{thm:rest}]
	\textbf{Proof for Path:} Let the set of vertices in the input path \PP be $\VV=\{v_1,\ldots,v_n\}$ and the set of edges $\EE=\{\{v_i,v_{i+1}\}:i\in[n-1]\}$. Note that the possible values of $n_v$ for any vertex $v$ in \PP is $0, 1$ and $2$. We show that there is a PSNE in \PP for all possible best response strategies. Let $S_i:=\beta(i)$ (since the game is fully homogeneous, the best-response function is the same for all players) be the set of the best responses of a player $v$ if $n_v=i$. Let $x_v$ be the response of a player $v\in \VV$. 
	
	\begin{itemize}
		
		\item If $0\in S_0$, then $(x_v=0)_{v\in\VV}$ forms a PSNE as clearly no player would deviate as $n_v=0$ for every player. So, allow us to assume for the rest of the proof that $S_0=\{1\}$.
		
		\item If $1 \in S_i, \forall i\in\{1,2\}$, then $(x_v=1)_{v\in\VV}$ forms a PSNE as the best response is 1 irrespective of the value of $n_v$ and hence no player would deviate. So, allow us to assume that we have either $1\notin S_1$ or $1\notin S_2$.
		
		\item If $0 \in S_1, 0 \in S_2$, then $((x_{v_i}=0)_{i \equiv1\pmod 2}, (x_{v_i}=1)_{i \equiv0\pmod 2})$ forms a PSNE.  If $i$ is an odd integer then $n_{v_i}>0$ and in this case one of the best responses is 0 and hence $v_i$ does not deviate.  If $i$ is an even integer then $n_{v_i}=0$ and in this case one of the best responses is $1$ (recall $S_0=\{1\}$) and hence $v_i$ does not deviate.
		
		\item If $1 \in S_1, 0 \in S_2$, then $((x_{v_i}=1)_{i \equiv1\pmod 2, i\ne n}, (x_{v_i}=0)_{i \equiv0\pmod 2,i\ne n}, x_{v_n}=1)$ forms a PSNE. If $i$ is an odd integer and not equal to $n$, then $n_{v_i}\leq 1$ and in this case one of the best responses  is 1 and hence $v_i$ does not deviate. If $i$ is an even number and not equal to $n$ then $n_{v_i}=2$ and in this case one of the best responses  is 0 and hence $v_i$ does not deviate. Note that $n_{v_n}\leq 1$ and hence in this case one of the best responses  is 1 and hence $v_n$ does not deviate.
		\item If $0 \in S_1, 1 \in S_2$. In this we have two sub-cases:
		\begin{itemize}
			\item \textbf{n is a multiple of $3$}: In this sub-case, $((x_{v_i}=0)_{i \not\equiv 2\pmod 3}, (x_{v_i}=1)_{i \equiv2\pmod 3})$ forms a PSNE. If we have $i \not\equiv 2 \pmod 3$, then $n_{v_i}=1$ and in this case, $0$ is a best response and hence $v_i$ does not deviate. If $i\equiv2\pmod 3$, then $n_{v_i}=0$ and in this case one of the best responses is $1$ and hence $v_i$ does not deviate. 
			\item \textbf{n is not a multiple of 3}: In this sub-case, $((x_{v_i}=0)_{i \not\equiv 1\pmod 3}, (x_{v_i}=1)_{i \equiv1\pmod 3})$ forms a PSNE.  If $i \not\equiv 1\pmod 3$, then $n_{v_i}=1$ and in this case one of the best responses is $0$ and hence $v_i$ does not deviate. If $i \equiv 1\pmod 3$, then $n_{v_i}=0$ and in this case one of the best responses is 1 and hence $v_i$ does not deviate. 
		\end{itemize}
	\end{itemize} 
	Since we have a PSNE for every possible best-response function, we conclude that there is always a PSNE in a fully homogeneous BNPG game on paths. Also, we can find a PSNE in paths in $\OO(n)$ time.
	
	\textbf{Proof for Complete graph:} We assume that the input graph \GG(\VV,\EE) is a complete graph.Let the utility function for all the players $v\in \VV$ be $U(x_v,n_v)=g(x_v+n_v)-c\cdot x_v$. If $\Delta g(n-1)\geq c$, then $(x_v=1)_{v\in \VV}$ is a PSNE. If $\Delta g(0)\leq c$, then $(x_v=0)_{v\in \VV}$ is a PSNE. If $\Delta g(n-1)< c$ and $\Delta g(0)> c$, then there should exist a $0<k\leq n-1$ such that $\Delta g(k)\leq c$ and $\Delta g(k-1)\geq c$ otherwise both $\Delta g(n-1)< c$ and $\Delta g(0)> c$ can't simultaneously hold true. Now we claim that if there exists a $0<k\leq n-1$ such that $\Delta g(k)\leq c$ and $\Delta g(k-1)\geq c$, then choosing any $k$ players and making their response 1 and rest of players response as $0$ would be PSNE. Any player $w$ whose response is $1$ has $n_{w}=k-1$ and since  $\Delta g(k-1)\geq c$, $w$ does not have any incentive to deviate.  Similarly, any player $w^\pr$ whose response is $0$ has $n_{w^\pr}=k$ and since  $\Delta g(k)\leq c$, $w^\pr$ does not have any incentive to deviate. This concludes the proof of the theorem as we showed that there is a PSNE in all possible cases. Also, clearly we can find a PSNE in \GG in $\OO(n)$ time.
	
	\textbf{Proof for Cycles:} We assume that the input graph is a Cycle. Let the set of vertices in the input cycle \CC be $\VV=\{v_1,\ldots,v_n\}$ and the set of edges $\EE=\{\{v_i,v_{i+1}\}:i\in[n-1]\}\cup \{v_n,v_1\}$. Note that the possible values of $n_v$(number of neighbours of $v$ choosing 1) for any vertex $v \in \CC$ is $0, 1,$ and $2$. We show that there is a PSNE in \CC for all possible best response strategies. Let $S_i:=\beta(i)$ (since the game is fully homogeneous, the best-response function is the same for all players) denote the set of the best responses of a player $v$ if $n_v=i$. Let $x_v$ denote the response of a player $v\in \VV$. 
	\begin{itemize}
		\item If $0\in S_0$ then $x_v=0$ for every player $v \in \VV$ forms a PSNE as clearly no player would deviate as $n_v=0$ for every player. So, allow us to assume that $S_0=\{1\}$ in the rest of the proof.
		\item If $1 \in S_2$, then $x_v=1$ for every player $v$ in \VV forms a PSNE as clearly no player would deviate as $n_v=2$ for every player. So, allow us to assume that $S_2=\{0\}$ in the rest of the proof.
		\item If $0 \in S_1$, then $((x_{v_i}=0)_{i \equiv1\pmod 2},(x_{v_i}=1)_{i \equiv0\pmod 2})$ forms a PSNE.  If $i$ is odd then $n_{v_i}>0$ and in this case one of the best responses is $0$ and hence $v_i$ does not deviate.  If $i$ is even number then $n_{v_i}=0$ and in this case one of the best responses is $1$ (recall, we have $S_0=\{1\}$) and hence $v_i$ does not deviate.
		\item If $1\in S_1$, then $((x_{v_i}=1)_{i \equiv1\pmod 2}, (x_{v_i}=0)_{i \equiv0\pmod 2})$ forms a PSNE. If $i$ is an odd number then $n_{v_i}\leq 1$ and in this case one of the best responses  is $1$ and hence $v_i$ does not deviate. If $i$ is an even number and then $n_{v_i}=2$ and in this case one of the best responses  is 0 and hence $v_i$ does not deviate.
	\end{itemize} 
	Since we have a PSNE for every possible best-response functions, we conclude that there is always a PSNE in a fully homogeneous BNPG game on cycles. Also, we can find a PSNE in cycles in $\OO(n)$ time.
	
	\textbf{Proof for Bicliques:} Let the input graph $\GG=(\VV,\EE)$ be a biclique; \VV is partitioned into $2$ sets namely $\VV_1=\{u_1,\ldots,u_{n_1}\}$ and $\VV_2=\{v_1,\ldots,v_{n_2}\}$ where $n_1+n_2=n$ and $\EE=\{(u_i,v_j):i\in[n_1],j\in[n_2]\}$. We show that there is a PSNE in \CC for all possible best response strategies. Let $S_i:=\beta(i)$ (since the game is fully homogeneous, the best-response function is the same for all players) denote the set of the best responses of a player $v$ if $n_v=i$. Let $x_v$ denote the response of a player $v\in \VV$. 
	\begin{enumerate}
		
		\item \textbf{If $n_1=n_2$}: For this case we have the following sub-cases: 
		
		\begin{itemize}
			\item If $0\in S_0$, then $(x_v=0)_{v\in\VV}$ forms a PSNE as clearly no player would deviate as $n_v=0$ for every player. So, allow us to assume that $S_0=\{1\}$.
			\item If $1\in S_{n_1}$, then $(x_v=1)_{v\in\VV}$ forms a PSNE as clearly no player would deviate as $n_v=n_1$ for every player. So, allow us to assume that $S_{n_1}=\{0\}$. However, then $((x_{v}=0)_{v\in\VV_1}, (x_v=1)_{v\in\VV_2})$ forms a PSNE.  If $v\in \VV_2$ then $n_{v}=0$ and in this case one of the best responses is 1 and hence $v$ won't deviate. If $v\in \VV_1$ then $n_{v}=n_1$ and in this case one of the best responses is 0 and hence $v$ won't deviate.
		\end{itemize}
		
		\item \textbf{If $n_1\ne n_2$}: For this case we have the following sub-cases:
		
		\begin{itemize}
			\item If $0\in S_0$ then $x_v=0$ for every player $v$ in \VV forms a PSNE as clearly no player would deviate as $n_v=0$ for every player. So, allow us to assume that $S_0=\{1\}$ for the rest of the proof.
			
			\item If $1\in S_{n_1}$, $1\in S_{n_2}$, then $(x_v=1)_{v\in\VV}$ forms a PSNE. If $v\in \VV_1$, then $n_{v}=n_2$ and in this case one of the best responses is $1$ and hence $v$ does not deviate. If $v\in \VV_2$, then $n_{v}=n_1$ and in this case one of the best responses is 1 and hence $v$ does not deviate.
			
			\item  If $0\in S_{n_1},1\in S_{n_2} $, then $((x_v=1)_{:v\in \VV_1}, (x_v=0)_{v\in \VV_2})$ forms a PSNE. If $v\in \VV_1$ then $n_{v}=0$ and in this case one of the best responses is $1$ (recall $S_0=\{1\}$) and hence $v$ does not deviate. If $v\in \VV_2$ then $n_{v}=n_1$ and in this case one of the best responses is $0$ and hence $v$ does not deviate.
			
			\item  If $0\in S_{n_1},0\in S_{n_2} $, then $((x_v=1)_{v\in \VV_1}, (x_v=0)_{v\in \VV_2})$ forms a PSNE. If $v\in \VV_1$ then $n_{v}=0$ and in this case one of the best responses is $1$ (recall $S_0=\{1\}$) and hence $v$ does not deviate. If $v\in \VV_2$ then $n_{v}=n_1$ and in this case one of the best responses is 0 and hence $v$ does not deviate.
			
			\item  If $1\in S_{n_1},0\in S_{n_2} $, then $((x_v=0)_{:v\in \VV_1}, (x_v=1)_{v\in \VV_2})$ forms a PSNE. If $v\in \VV_1$ then $n_{v}=n_2$ and in this case one of the best responses is $0$ and hence $v$ does not deviate. If $v\in \VV_2$ then $n_{v}=0$ and in this case one of the best responses is $1$ (recall $S_0=\{1\}$) and hence $v$ does not deviate.
		\end{itemize}
	\end{enumerate}
	Since we have a PSNE for every possible best-response functions, we conclude that there is always a PSNE in a fully homogeneous BNPG game on biclique. Also, we can find a PSNE in biclique in $\OO(n)$ time.
\end{proof}
\end{document}